\documentclass[preprint,numbers]{llncs}

\usepackage{comment}
\usepackage{amsmath}
\usepackage{paralist,multirow}
\usepackage{amsfonts, amssymb}
\usepackage{thmtools, thm-restate}
\usepackage{tikz}
\usepackage{stmaryrd}
\usepackage{hyperref}
\usepackage{subfigure}
\usetikzlibrary{decorations.pathreplacing}
\usepackage{microtype}

\usetikzlibrary{shapes,decorations,er,positioning}
\usepackage{listings}

\usepackage{xifthen}
\newboolean{conferenceversion}
\setboolean{conferenceversion}{false}

\DeclareMathOperator*{\argmin}{argmin}
\newcommand{\may}{\mathit{May}}
\newcommand{\must}{\mathit{Must}}
\newcommand{\existsh}{\mathit{EH}}
\newcommand{\existsm}{\mathit{EM}}
\newcommand{\mydots}\dots

\newcommand{\sizefactor}{.9}

\begin{document}

\spnewtheorem{notation}{Notation}{\bfseries}{\itshape}

\newcommand{\soft}[1]{\textsc{#1}}
\newcommand{\benchmarkname}[1]{\textsf{#1}}
\newcommand{\todo}[1]{{\color{red}[#1]}}
\newcommand{\redtodo}[1]{{\color{red}[#1]}}
   \renewcommand{\todo}[1]{}

\newcommand{\nhspace}{ \hspace{-0.0035mm}}
\newcommand{\nvspacesmall}{ \hspace{-0.0025mm}}

\author{
  Valentin Touzeau\inst{1} \and
  Claire Ma\"{\i}za\inst{1} \and
  David Monniaux\inst{1} \and
  Jan Reineke\inst{2}
}

\institute{
  Univ. Grenoble Alpes, VERIMAG, F-38000 Grenoble, France\\
  CNRS, VERIMAG, F-38000 Grenoble, France\\
  \email{firstname.lastname@univ-grenoble-alpes.fr}
  \and
  Saarland University, Saarland Informatics Campus\\
  Saarbr\"ucken, Germany\\
  \email{reineke@cs.uni-saarland.de}
}

\title{Ascertaining Uncertainty\\for Efficient Exact Cache Analysis
\thanks{This work was partially supported by the
\href{http://erc.europa.eu/}{European Research Council} under the
European Union's Seventh Framework Programme (FP/2007-2013) / ERC Grant
Agreement nr. 306595 \href{http://stator.imag.fr}{``STATOR''.}}
}

\maketitle

\begin{abstract}
Static cache analysis characterizes a program's cache behavior by determining in a sound but approximate manner which memory accesses result in cache hits and which result in cache misses.
Such information is valuable in optimizing compilers, worst-case execution time analysis, and side-channel attack quantification and mitigation.

Cache analysis is usually performed as a combination of ``must'' and ``may'' abstract interpretations, 
classifying instructions as  either ``always hit'', ``always miss'', or ``unknown''.
Instructions classified as ``unknown'' might result in a hit or a miss depending on program inputs or the initial cache state.
It is equally possible that they do in fact always hit or always miss, but the cache analysis is too coarse to see it.

Our approach to eliminate this uncertainty consists in
\begin{inparaenum}[(i)]
\item a novel abstract interpretation able to ascertain that a particular instruction may definitely cause a hit and a miss on different paths, and 
\item an exact analysis, removing all remaining uncertainty, based on model checking, using abstract-interpretation results to prune down the model for scalability.
\end{inparaenum}

We evaluated our approach on a variety of examples;
it notably improves precision upon classical abstract interpretation
at reasonable cost.
\end{abstract}



\section{Introduction}
There is a large gap between processor and memory speeds termed the ``memory wall''~\cite{Wulf1995}.
To bridge this gap, processors are commonly equipped with caches, i.e., small but fast on-chip memories that hold recently-accessed data, in the hope that most memory accesses can be served at a low latency by the cache instead of being served by the slow main memory.
Due to temporal and spatial locality in memory access patterns caches are often highly effective.


In hard real-time applications, it is important to bound a program's \emph{worst-case execution time} (WCET). For instance, if a control loop runs at 100~Hz, one must show that its WCET is less than 0.01~s.
In some cases, measuring the program's execution time on representative inputs and adding a safety margin may be enough, but in safety-critical systems one may wish for a higher degree of assurance and use static analysis to cover all cases.
On processors with caches, such a static analysis involves classifying memory accesses into cache hits, cache misses, and unclassified~\cite{Wilhelm_et_al_TECS08}.
Unclassified memory accesses that in reality result in cache hits may lead to gross overestimation of the WCET.

Tools such as \soft{Otawa}%
\footnote{\url{http://www.otawa.fr/}: an academic tool developed at IRIT, Toulouse.}
and \soft{aiT}%
\footnote{\url{https://www.absint.com/ait/}: a commercial tool developed by Absint GmbH.}
compute an upper bound on the WCET of programs after first running a static analysis based on abstract interpretation~\cite{Lv16} to classify memory accesses.
Our aim, in this article, is to improve upon that approach with a refined abstract interpretation and a novel encoding into finite-state model checking.

Caches may also leak secret information~\cite{Canteau_et_al_RR5881_2006} to other programs running on the same machine---through the shared cache state---or even to external devices---due to cache-induced timing variations.
For instance, cache timing attacks on software implementations of the Advanced Encryption Standard \cite{Bernstein_2005} were one motivation for adding specific hardware support for that cipher to the x86 instruction set~\cite{Mowery:2012:AXC:2381913.2381917}.
Cache analysis may help identify possibilities for such \emph{side-channel attacks} and quantify the amount of information leakage~\cite{Doychev2015};
improved precision in cache analysis then translates into fewer false alarms and tighter bounds on leakage.

An ideal cache analysis would statically classify every memory access at every machine-code instruction in a program into one of three cases:
\begin{inparaenum}[i)]
\item
the access is a cache hit in all possible executions of the program
\item
the access is a cache miss in all possible executions of the program
\item
in some executions the access is a hit and in others it is a miss.
\end{inparaenum}
However, no cache analysis can perfectly classify all accesses into these three categories.

\newcommand{\versatz}{13mm}

\begin{figure}[t]
\begin{center}
\begin{tikzpicture}[scale=0.9, transform shape, inner sep=2mm, 
				maymust/.style={rectangle,draw=blue!70,fill=blue!30,thick}, 
				defunknown/.style={rectangle,draw=red!45,fill=red!15,thick},
				finalresult/.style={rectangle,draw=green!50, fill=green!20, thick}]


    \draw[finalresult] (2.7, -1.5) rectangle (12.2, -2.5);
    \node at (10.75, -2) {Result after MC};

    \node [maymust] (top) at (6, 0) {unknown};
    \node [defunknown, below of=top, xshift=\versatz] (existsmiss) {$\exists$Miss}; 
    \node [defunknown, below of=top, xshift=-\versatz] (existshit) {$\exists$Hit}; 
    \node [maymust, below of=existsmiss, xshift=\versatz] (allmiss) {$\forall$Miss}; 
    \node [defunknown, below of=existsmiss, xshift=-\versatz] (defunknown) {$\exists$Hit $\wedge$ $\exists$Miss}; 
    \node [maymust, below of=existshit, xshift=-\versatz] (allhit) {$\forall$Hit}; 
    
    \node [maymust, right of=top, xshift=3.15*\versatz, yshift=-2mm] (legend1) {Classical AI};
    \node [defunknown, below of=legend1, yshift=3mm] {Our new AI};
    \node [above of=legend1, yshift=-5mm, xshift=-5mm] {Legend:};

    \draw [thick] (top) -- (existsmiss);
    \draw [thick] (top) -- (existshit);
    \draw [thick] (existshit) -- (allhit);
    \draw [thick] (existshit) -- (defunknown);
    \draw [thick] (existsmiss) -- (allmiss);
    \draw [thick] (existsmiss) -- (defunknown);
\end{tikzpicture}
\end{center}

\caption{Possible classifications of classical abstract-interpretation-based cache analysis, our new abstract interpretation, and after refinement by model checking.}\label{fig:classifications}
\end{figure}
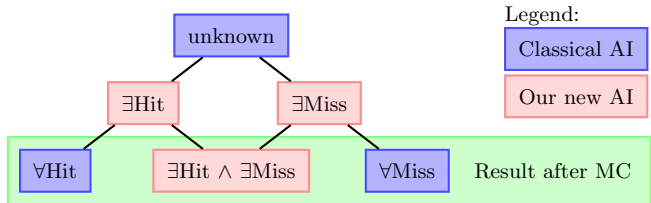

%
%

One first reason is that perfect cache analysis would involve testing the reachability of individual program statements, which is undecidable.%
\footnote{One may object that given that we consider machine-level aspects, memory is bounded and thus properties are decidable. The time and space complexity is however prohibitive.}
A simplifying assumption often used, including in this article, is that all program paths are feasible---this is safe, since it overapproximates possible program behaviors.
Even with this assumption, analysis is usually performed using sound but incomplete abstractions that can safely determine that some accesses always hit (``$\forall$Hit'' in Figure~\ref{fig:classifications}) or always miss (``$\forall$Miss'' in Fig. \ref{fig:classifications}). The corresponding analyses are called \emph{may} and \emph{must} analysis and referred to as ``classical AI'' in Fig.~\ref{fig:classifications}.
Due to incompleteness the status of other accesses however remains ``unknown''~(Fig.~\ref{fig:classifications}).

\subsubsection{Contributions}
In this article, we propose an approach to eliminate this uncertainty, with two main contributions (colored red and green in Figure~\ref{fig:classifications}):
\begin{enumerate}
\item\label{item:contribution:abstract} A novel abstract interpretation that safely concludes that certain accesses are {hits in some} executions (``$\exists$Hit''), {misses in some} executions (``$\exists$Miss''), or {hits in some and misses in other} executions (``$\exists$Hit $\wedge$ $\exists$Miss'' in Fig. \ref{fig:classifications}).
Using this analysis and prior must- and may- cache analyses, most accesses are precisely classified.
\item\label{item:contribution:mc} The classification of accesses with remaining uncertainty (``unknown'', ``$\exists$Hit'', and ``$\exists$Miss'') is refined by model checking using an exact abstraction of the behavior of the cache replacement policy. The results from the abstract interpretation in the first analysis phase are used to dramatically reduce the complexity of the model. 
\end{enumerate}

Because the model-checking phase is based on an exact abstraction of the cache replacement policy, our method, overall, is \emph{optimally precise}: it answers precisely whether a given access is always a hit, always a miss, or a hit in some executions and a miss in others (see ``Result after MC'' in Fig. \ref{fig:classifications}).\footnote{This completeness is relative to an execution model where all control paths are feasible, disregarding the functional semantics of the edges.}
This precision improvement in access classifications can be beneficial for tools built on top of the cache analysis: in the case of WCET analysis for example, a precise cache analysis  not only improves the computed WCET bound; it can also lead to a faster analysis. Indeed, in case of an unclassified access, both possibilities (cache hit and cache miss) have to be considered~\cite{Lundqvist99,Reineke06}.  %

The model-checking phase would be sufficient to resolve all accesses, but our experiments show this does not scale; it is necessary to combine it with the abstract-interpretation phase for tractability, thereby reducing \begin{inparaenum}[(a)]
\item the number of model-checker calls, and
\item the size of each model-checking problem.
\end{inparaenum}

%

\section{Background: Caches and Static Cache Analysis}
\label{sec:cache_modeling}
\subsubsection{Caches}

Caches are fast but small memories that store a subset of the main
memory's contents to bridge the latency gap between the CPU and main
memory.  To profit from spatial locality and to reduce management
overhead, main memory is logically partitioned into a set of
\emph{memory blocks}~M. 
Each block is cached as a whole in a cache line of the same size.

When accessing a memory block, the cache logic has to determine
whether the block is stored in the cache (``cache hit'') or not
(``cache miss''). 
For efficient look up, each block can only
be stored in a small number of cache lines known as a \emph{cache set}.
Which cache set a memory block maps to is determined by a subset of the bits of its address.
The cache is partitioned into equally-sized cache sets.  The size~$k$ of a
cache set in blocks is called the \emph{associativity} of the cache.

Since the cache is much smaller than main memory, a \emph{replacement policy} must decide which memory block to replace upon a cache miss.
Importantly, replacement policies treat sets independently\footnote{To our best knowledge, the only exception to this rule is the \emph{pseudo round-robin} policy, found, e.g., in the ARM Cortex A-9.}, so that accesses to one set do not influence replacement decisions in other sets. 
Well-known replacement policies are
least-recently-used (LRU), used, e.g., in various Freescale processors such as the MPC603E and the TriCore17xx; pseudo-LRU (PLRU), a cost-efficient variant of LRU; and first-in first-out (FIFO).  
In this article we focus exclusively on LRU. The application of our ideas to other policies is left as future work. 

LRU naturally gives rise to a notion of \emph{ages} for memory blocks:
The age of a block $b$ is the number of pairwise different blocks that map to the same cache set as $b$ that have been accessed since the last access to~$b$.
If a block has never been accessed, its age is $\infty$.
Then, a block is cached if and only if its age is less than the cache's associativity $k$.

Given this notion of ages, the state of an LRU cache can be modeled by a mapping that assigns to each memory block its age, where ages are truncated at $k$, i.e., we do not distinguish ages of uncached blocks.
We denote the set of cache states by $C = M \rightarrow \{0, \dots, k\}$.
Then, the effect of an access to memory block~$b$ under LRU replacement can be formalized as follows\footnote{Assuming for simplicity that all cache blocks map to the same cache set.}:
	{\small\begin{align}
		update : C \times M & \rightarrow C\nonumber\\
		(q, b)	& \mapsto \lambda b'.\label{eq:lruconcreteupdate}
		\begin{cases}
			0			& \text{ if } b' = b\\
			q(b')			& \text{ if } q(b') \geq q(b)\\
			q(b')+1		& \text{ if } q(b') < q(b) \wedge q(b') < k\\
			k			& \text{ if } q(b') < q(b) \wedge q(b') = k
		\end{cases}
	\end{align}}

\subsubsection{Programs as Control-flow Graphs}

As is common in program analysis and in particular in work on cache analysis, we abstract the program under analysis by its control-flow graph: 
vertices represent control locations and edges represent the possible flow of control through the program.
In order to analyze the cache behavior, edges are adorned with the addresses of the memory blocks that are accessed by the instruction, including the instruction being fetched. 

For instruction fetches in a program without function pointers or computed jumps, this just entails knowing the address of every instruction---thus the program must be linked with absolute addresses, as common in embedded code.
For data accesses, a pointer analysis is required to compute a set of possible addresses for every access.
If several memory blocks may be alternatively accessed by an instruction, multiple edges may be inserted; so there may be multiple edges between two nodes.
We therefore represent a control-flow graph by a tuple $G = (V, E)$, where $V$ is the set of vertices and $E \subseteq V \times (M \cup \{\bot\}) \times V$ is the set of edges, where $\bot$ is used to label edges that do not cause a memory access.

The resulting control-flow graph $G$ does not include information on the functional semantics of the instructions, e.g. whether they compute an addition.
All paths in that graph are considered feasible, even if, taking into account the instruction semantics, they are not---e.g. a path including the tests $x \leq 4$ and $x \geq 5$ in immediate succession is considered feasible even though the two tests are mutually exclusive.
All our claims of completeness are relative to this model.

As discussed above, replacement decisions for a given cache set are usually independent of memory accesses to other cache sets.
Thus, analyzing the behavior of $G$ on all cache sets is equivalent to separately analyzing its projections onto individual cache sets: a projection of $G$ on a cache set $S$ is $G$ where only blocks mapping to $S$ are kept.
Projected control-flow graphs may be simplified, e.g. a self-looping edge labeled with no cache block may be removed.
Thus, we assume in the following that the analyzed cache is fully associative, i.e. of a single cache set.

\subsubsection{Collecting Semantics}
In order to classify memory accesses as ``always hit'' or ``always miss'', cache analysis needs to characterize for each control location in a program \emph{all} cache states that may reach that location in any execution of the program. This is commonly called the \emph{collecting semantics}.

Given a control-flow graph $G=(V, E)$, the \emph{collecting semantics} is defined as the least solution to the following set of equations, where $R^C : V \rightarrow \mathcal{P}(C)$ denotes the set of reachable concrete cache configurations at each program location, and $R^C_0(v)$ denotes the set of possible initial cache configurations:
\begin{equation}\label{def:concrete_reachable}
 	\forall v' \in V: R^C(v') = R^C_0(v') \cup \bigcup_{(v,b, v') \in E} update^C(R^C(v), b),
\end{equation}
where $update^C$ denotes the cache update function lifted to sets of states, i.e., $update^C(Q, b) = \{update(q, b) \mid q \in Q\}$.

Explicitly computing the collecting semantics is practically infeasible. 
For a tractable analysis, it is necessary to operate in an abstract domain whose elements compactly represent large sets of concrete cache states.

\subsubsection{Classical Abstract Interpretation of LRU Caches}
To this end, the classical abstract interpretation of LRU caches \cite{Ferdinand99} assigns to every memory block at every program location an interval of ages enclosing the possible ages of the block during any program execution.
The analysis for upper bounds, or \emph{must analysis}, can prove that a block must be in the cache;
conversely, the one for lower bounds, or \emph{may analysis}, can prove that a block may not be in the cache.

The domains for abstract cache states under may and must analysis are $\mathcal{A}_{\may} = \mathcal{A}_{\must} = C = M \rightarrow \{0,...,k\}$, where ages greater than or equal to the cache's associativity~$k$ are truncated at $k$ as in the concrete domain. 
For reasons of brevity, we here limit our exposition to the must analysis. 
The set of concrete cache states represented by abstract cache states is given by the concretization function:
		$\gamma_{\must}(\hat{q}_{\must}) = \{q \in C \mid \forall m \in M: q(m) \leq \hat{q}_{\must}\}.$
Abstract cache states can be joined by taking their pointwise maxima: $\hat{q}_{M1} \sqcup_{\must} \hat{q}_{M2} = \lambda m \in M: \max \{\hat{q}_{M1}(m), \hat{q}_{M2}(m)\}$.
For reasons of brevity, we also omit the definition of the abstract transformer ${update}_{\must}$, which closely resembles its concrete counterpart given in (\ref{eq:lruconcreteupdate}), and which can be found e.g. in~\cite{Reineke_PhD}.

Suitably defined abstract semantics $R_{\must}$ and $R_{\may}$ can be shown to overapproximate their concrete counterpart: 
\begin{theorem}[Analysis Soundness~\cite{Ferdinand99}]
	The may and the must abstract semantics are safe approximations of the collecting semantics:
	\begin{equation}
		\forall v \in V: R^C(v) \subseteq \gamma_{\must}(R_{\must}(v)), R^C(v) \subseteq \gamma_{\may}(R_{\may}(v)).
	\end{equation}
\end{theorem}

\section{Abstract Interpretation for Definitely Unknown}
\label{sec:abstract_interpretation}
\emph{All proofs can be found in \ifthenelse{\boolean{conferenceversion}}{Appendix A of the technical report~\cite{technical_report}}{Appendix~\ref{sec:abstract_interpretation_proofs}}.} %
\begin{figure}[tb]
%
\begin{center}
\begin{tikzpicture}[scale=0.85,transform shape,->,node distance=1cm,
		defunknown/.style={rectangle,draw=red!45,fill=red!15,thick},
		finalresult/.style={rectangle,draw=green!50, fill=green!20, thick},
		programmodel/.style={rectangle,draw=blue!50, fill=blue!20, thick},
		cachemodel/.style={rectangle,draw=blue!50, fill=blue!20, thick}]

  \node[defunknown, entity] (node1) {
  \begin{tabular}{c}
   Abstract Interpretation \\ \hline
   may/must analysis \\
   $\exists$hit/$\exists$miss analysis
  \end{tabular}

  };
 \node[programmodel, attribute] (cfg) [left = of node1, xshift=8mm, yshift=-1cm, text width=2cm, align=center] {Control-flow graph};
 \node[cachemodel, attribute] (cc) [left = of node1, xshift=8mm, yshift=1cm, text width=2cm, align=center] {Cache configuration};
  \node[programmodel, attribute] (pm) [right = of node1, xshift=-8mm, yshift=-1cm, text width=2.2cm, align=center] {Simplified program model};
  \node[cachemodel, attribute] (cm) [right = of node1, xshift=-8mm, yshift=1cm, text width=2.2cm, align=center] {Focused cache model};
  \node[finalresult, entity] (mc) [right = of node1, xshift=2.6cm, text width = 1.2cm, text centered, minimum height=1cm] {Model checker};

  \draw [-, decorate,thick, decoration={brace,amplitude=10pt},xshift=-4pt,yshift=0pt] (-5.1,1.6) -- (2.0,1.6) node [black,midway,yshift=0.6cm] {Section \ref{sec:abstract_interpretation}};
  \draw [-, decorate,thick, decoration={brace,amplitude=10pt},xshift=-4pt,yshift=0pt] (2.1,1.6) -- (7.2,1.6) node [black,midway,yshift=0.6cm] {Section \ref{sec:model_checking}};

  \path (cfg) edge (node1)
  (cc)edge (node1)
  (node1)edge(pm)
  (node1)edge(cm)
  (cm)edge(mc)
  (pm)edge(mc);

\end{tikzpicture}
\end{center}

\caption{Overall analysis flow.}\label{fig_AI+MC}
\end{figure}
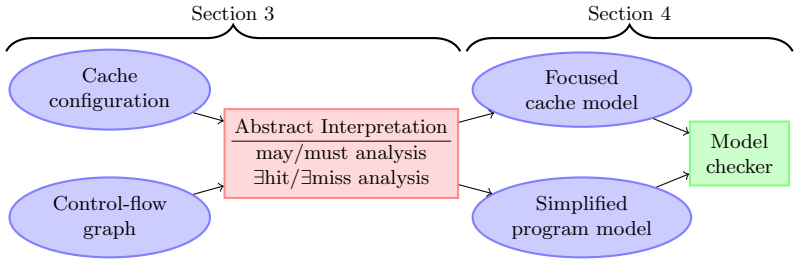
Together, may and must analysis can classify accesses as ``always hit'', ``always miss'' or ``unknown''.
An access classified as ``unknown'' may still be ``always hit'' or ``always miss'' but not detected as such due to the imprecision of the abstract analysis;
otherwise it is ``definitely unknown''.
Properly classifying ``unknown'' blocks into ``definitely unknown'', ``always hit'', or ``always miss'' using a model checker is costly.
We thus propose an abstract analysis that safely establishes that some blocks are ``definitely unknown'' under LRU replacement.

Our analysis steps are summarized in Figure~\ref{fig_AI+MC}.
Based on the control-flow graph and on an initial cache configuration, the abstract-interpretation phase classifies some of the accesses as ``always hit'', ``always miss'' and ``definitely unknown''.
Those accesses are already precisely classified and thus do not require a model-checking phase. The AI phase thus reduces the number of accesses to be classified by the model checker. In addition, the results of the AI phase are used to simplify the model-checking phase, which will be discussed in detail in Section~\ref{sec:model_checking}.

An access is ``definitely unknown'' if there is a concrete execution in which the access misses and another in which it hits.
The aim of our analysis is to prove the existence of such executions to classify an access as ``definitely unknown''.
Note the difference with classical may/must analysis and most other abstract interpretations, which compute properties that hold \emph{for all executions}, while here we seek to prove that \emph{there exist} two executions with suitable properties.

An access to a block $a$ results in a hit if $a$ has been accessed recently, i.e., $a$'s age is low.
Thus we would like to determine the minimal age that $a$ may have in a reachable cache state immediately prior to the access in question.     
The access can be a hit if and only if this minimal age is lower than the cache's associativity.
Because we cannot efficiently compute exact minimal ages, we devise an \emph{Exists Hit} (EH) analysis to compute safe upper bounds on minimal ages.
Similarly, to be sure there is an execution in which accessing $a$ results in a miss, we compute a safe lower bound on the maximal age of $a$ using the \emph{Exists Miss} (EM) analysis.

\lstset{
    showlines=true,
    columns=fullflexible,
    breaklines=true,
    captionpos=b,
    basicstyle=\sffamily\normalsize,
    commentstyle=\itshape,
    keywordstyle=\bfseries,
    numbers=none,
    mathescape=true,
    escapechar=@,
}

\lstdefinestyle{alang}{
    morekeywords={if,then,else,end,while,assert,skip,true,false},
    numbers=none,
	literate=
	{:=}{$\coloneqq\ $}1
	{==}{$=\ $}1
	{!=}{$\ne\ $}1
	{&&}{$\wedge\ $}1
	{||}{$\vee\ $}1
	{<=}{$\le\ $}1
	{>=}{$\ge\ $}1
}

\tikzset{pp/.style={circle, draw, fill=black, inner sep=0pt, minimum size=6pt, text centered}}
\tikzset{cfge/.style={draw,-stealth,thick}}
\tikzset{cfgep/.style={draw,,thick}}
\tikzset{cfgfe/.style={draw,-stealth,ultra thick}}
\tikzset{cfgp/.style={
     style=cfge,
     decoration={snake, amplitude=.4mm, segment length=2mm, post length=1mm},
     decorate
 }}
\tikzset{cfgbe/.style={style=cfge, dashed}}

\def\ppsimpleedge#1#2{\draw[cfge] (#1) -- (#2);}
\def\ppedge#1#2#3{\draw[cfge] (#1) -- (#2) node[midway,draw,fill=white,rounded corners=4pt,minimum height=16pt,minimum width=16pt] {#3};}
\def\ppedgebend#1#2#3#4{\draw[cfge] (#1) -- (#2) node[draw,fill=white,rounded corners=4pt,minimum height=16pt,minimum width=16pt] {#4} -- (#3);}

\newcommand{\skalierung}{0.9}

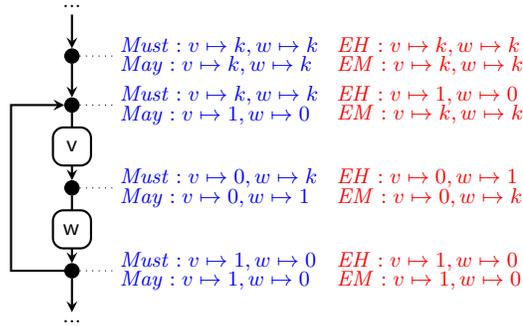
\begin{figure}[t]
\begin{center}
	\begin{tikzpicture}[scale=\skalierung,transform shape]
            \node (a0) {...};
            \node[pp] (a) [below=5mm of a0] {};
            \node[pp] (b) [below=5mm of a] {};
            \node[pp] (c) [below=10mm of b] {};
            \node[pp] (d) [below=10mm of c] {};
            \node (i1) [below=4.5mm of c, xshift = -4mm] {};
            \node (i2) [below=4.5mm of c, xshift = 4mm] {};
            \node (e) [below=5mm of d] {...};
        
            \ppsimpleedge{a0}{a}
            \ppsimpleedge{a}{b}
            \ppedge{b}{c}{\lstinline[style=alang]|v|}
            \ppedge{c}{d}{\lstinline[style=alang]|w|}
            \draw[cfge] (d) -| ([xshift=-9mm] c.center) |- (b);
            \ppsimpleedge{d}{e}
        

            \node[blue,right=5mm of a, yshift=0mm, text width=4.25cm] (blaa) {$\must: v \mapsto k, w \mapsto k$\\[-1mm] $\may: v \mapsto k, w \mapsto k$};
            \node[red,right=37mm of a, yshift=0mm, text width=4.25cm] {$\existsh: v \mapsto k, w \mapsto k$\\[-1mm] $\existsm: v \mapsto k, w \mapsto k$};
            \node[blue,right=5mm of b, yshift=0mm, text width=4.25cm] (blab) {$\must: v \mapsto k, w \mapsto k$\\[-1mm] $\may: v \mapsto 1, w \mapsto 0$};
            \node[red,right=37mm of b, yshift=0mm, text width=4.25cm] {$\existsh: v \mapsto 1, w \mapsto 0$\\[-1mm] $\existsm: v \mapsto k, w \mapsto k$};
            \node[blue,right=5mm of c, yshift=0mm, text width=4.25cm] (blac) {$\must: v \mapsto 0, w \mapsto k$\\[-1mm] $\may: v \mapsto 0, w \mapsto 1$};
            \node[red,right=37mm of c, yshift=0mm, text width=4.25cm] {$\existsh: v \mapsto 0, w \mapsto 1$\\[-1mm] $\existsm: v \mapsto 0, w \mapsto k$};
            \node[blue,right=5mm of d, yshift=0mm, text width=4.25cm] (blad) {$\must: v \mapsto 1, w \mapsto 0$\\[-1mm] $\may: v \mapsto 1, w \mapsto 0$};
            \node[red,right=37mm of d, yshift=0mm, text width=4.25cm] {$\existsh: v \mapsto 1, w \mapsto 0$\\[-1mm] $\existsm: v \mapsto 1, w \mapsto 0$};
            
	   \draw[dotted] (a) -- (blaa);
	   \draw[dotted] (b) -- (blab);
	   \draw[dotted] (c) -- (blac);
	   \draw[dotted] (d) -- (blad);
	\end{tikzpicture}
\end{center}
\nvspacesmall
\caption{Example of two accesses in a loop that are definitely unknown. May/Must and EH/EM analysis results are given next to the respective control locations.}\label{fig:defunknown}
\nvspacesmall
\end{figure}

\paragraph{Example.} Let us now consider a small example.
In Figure~\ref{fig:defunknown}, we see a small control-flow graph corresponding to a loop that repeatedly accesses memory blocks $v$ and~$w$.
Assume the cache is empty before entering the loop.
Then, the accesses to $v$ and $w$ are definitely unknown in fully-associative caches of associativity 2 or greater: they both miss in the first loop iteration, while they hit in all subsequent iterations.
Applying standard may and must analysis, both accesses are soundly classified as ``unknown''.
On the other hand, applying the EH analysis, we can determine that there are cases where $v$ and $w$ hit.
Similarly, the EM analysis derives that there exist executions in which they miss.
Combining those two results, the two accesses can safely be classified as definitely unknown.

\todo{might raise the question about loop peeling. should we counter that proactively? we could say that loop peeling eliminates this example, but others remain.}

\paragraph{}We will now define these analyses and their underlying domains more formally. 
The EH analysis maintains upper bounds on the minimal ages of blocks.
In addition, it includes a must analysis to obtain upper bounds on all possible ages of blocks, which are required for precise updates.
Thus the domain for abstract cache states under the EH analysis is
$\mathcal{A}_{\existsh} = (M \rightarrow \{0, \dots, k-1, k\}) \times \mathcal{A}_{\must}$.
Similarly, the EM analysis maintains lower bounds on the maximal ages of blocks and includes a regular may analysis:
$\mathcal{A}_{\existsm} = (M \rightarrow \{0, \dots, k-1, k\}) \times \mathcal{A}_{\may}$.
In the following, for reasons of brevity, we limit our exposition to the EH analysis.
The EM formalization is analogous and can be found in \ifthenelse{\boolean{conferenceversion}}{the technical report~\cite{technical_report}}{the appendix}.


The properties we wish to establish, i.e. bounds on minimal and maximal ages, are actually \emph{hyperproperties}~\cite{ClarksonCSF2008}: they are not properties of individual reachable states but rather of the entire \emph{set} of reachable states.
Thus, the conventional approach in which abstract states concretize to sets of concrete states that are a superset of the actual set of reachable states is {not applicable}.
Instead, we express the meaning, $\gamma_{\existsh}$, of abstract states by \emph{sets of sets} of concrete states.
A set of states $Q$ is represented by an abstract EH state $(\hat{q}, \hat{q}_{\must})$, if for each block~$b$, $\hat{q}(b)$ is an upper bound on $b$'s minimal age in $Q$, $\min_{q \in Q} q(b)$:
{\small\begin{align}
		\gamma_{\existsh}:	\mathcal{A}_{\existsh}		& \rightarrow	\mathcal{P}(\mathcal{P}(C))\nonumber\\
				(\hat{q}, \hat{q}_{\must})	& \mapsto 	\big\{Q \subseteq \gamma_{\must}(\hat{q}_{\must}) \mid \forall b \in M: \min_{q \in Q} q(b) \leq \hat{q}(b)\big\}
\end{align}}
The actual set of reachable states is an element rather than a subset of this concretization.
The concretization for the must analysis, $\gamma_{\must}$, is simply lifted to this setting.
Also note that---possibly contrary to initial intuition---our abstraction cannot be expressed as an underapproximation, as different blocks' minimal ages may be attained in different concrete states.

The abstract transformer $update_{\existsh}((\hat{q}_{\existsh}, \hat{q}_{\must}), b)$ corresponding to an access to block $b$ is the pair $(\hat{q}_{\existsh}',update_{\must}(\hat{q}_{\must}, b))$, where
{\small\begin{eqnarray}
\hat{q}_{\existsh}' = \lambda b'.
		\begin{cases}
			0		& \text{ if } b' = b\\
			\hat{q}(b')	& \text{ if } \hat{q}_{\must}(b) \leq \hat{q}(b')\\
			\hat{q}(b')+1	& \text{ if } \hat{q}_{\must}(b) > \hat{q}(b') \wedge \hat{q}(b') < k\\
			k		& \text{ if } \hat{q}_{\must}(b) > \hat{q}(b') \wedge \hat{q}(b') = k\\
		\end{cases}
\end{eqnarray}}
Let us explain the four cases in the transformer above.
After an access to~$b$, $b$'s age is $0$ in all possible executions. Thus, $0$ is also a safe upper bound on its minimal age (case 1). 
The access to $b$ may only increase the ages of younger blocks (because of the LRU replacement policy). 
In the cache state in which $b'$ attains its minimal age, it is either younger or older than $b$.
If it is younger, then the access to $b$ may increase $b'$'s actual minimal age, but not beyond $\hat{q}_{\must}(b)$, which is a bound on $b$'s age in every cache state, and in particular in the one where $b'$ attains its minimal age. Otherwise, if $b'$ is older, its minimal age remains the same and so may its bound. This explains why the bound on $b'$'s minimal age does not increase in case 2.
Otherwise, for safe upper bounds, in cases 3 and~4, the bound needs to be increased by one, unless it has already reached $k$.

\begin{restatable}[Local Consistency]{lemma}{ConsistencyUpdateEH}\label{lem:localconsistency}
The abstract transformer $update_{\existsh}$ soundly approximates its concrete counterpart $update^C$:
\begin{multline}
		\forall (\hat{q}, \hat{q}_{\must}) \in \mathcal{A}_{\existsh}, \forall b \in M, \forall Q \in \gamma_{\existsh}(\hat{q}, \hat{q}_{\must}):\\
 update^C (Q, b) \in \gamma_{\existsh}(update_{\existsh}((\hat{q}, \hat{q}_{\must}), b)).
\end{multline}%
\end{restatable}

How are EH states combined at control-flow joins? The standard must join can be applied for the must analysis component.
In the concrete, the union of the states reachable along all incoming control paths is reachable after the join.
It is thus safe to take the \emph{minimum} of the upper bounds on minimal ages:
\begin{equation}
(\hat{q}_1, \hat{q}_{\must1}) \sqcup_{\existsh} (\hat{q}_2, \hat{q}_{\must2}) =
(\lambda b. \min(\hat{q}_1(b), \hat{q}_2(b)), \hat{q}_{\must1} \sqcup_{\must} \hat{q}_{\must2})
\end{equation}

\begin{restatable}[Join Consistency]{lemma}{ConsistencyJoinEH}\label{lem:joinconsistency}
	The join operator $\sqcup_{\existsh}$ is correct:
	\begin{multline}
			\forall ((\hat{q}_1, \hat{q}_{M1}), (\hat{q}_2, \hat{q}_{M2})) \in \mathcal{A}_{\existsh}^2, Q_1 \in \gamma_{\existsh}(\hat{q}_1, \hat{q}_{M1}), Q_2 \in \gamma_{\existsh}(\hat{q}_2, \hat{q}_{M2}):\\
			Q_1 \cup Q_2 \in \gamma_{\existsh}((\hat{q}_1, \hat{q}_{M1}) \sqcup_{\existsh} (\hat{q}_2, \hat{q}_{M2})).
	\end{multline}
\end{restatable}

Given a control-flow graph $G=(V, E)$, the \emph{abstract EH semantics} is defined as the least solution to the following set of equations, where $R_{\existsh} : V \rightarrow \mathcal{A}_{\existsh}$ denotes the abstract cache configuration associated with each program location, and $R^C_0(v) \in \gamma_{\existsh}(R_{\existsh,0}(v))$ denotes the initial abstract cache configuration:
\begin{equation}
 	\forall v' \in V: R_{\existsh}(v') = R_{\existsh,0}(v') \sqcup_{\existsh} \bigsqcup_{(v, b, v') \in E} update_{\existsh}(R_{\existsh}(v), b).
\end{equation}
It follows from Lemmas~\ref{lem:localconsistency} and~\ref{lem:joinconsistency} that the abstract EH semantics includes the actual set of reachable concrete states: 
\begin{restatable}[Analysis Soundness]{theorem}{SoundnessEHEM}
	The abstract EH semantics includes the collecting semantics:
	$\forall v \in V: R^C(v) \in \gamma_{\existsh}(R_{\existsh}(v))$.
\end{restatable}
We can use the results of the EH analysis to determine that an access results in a hit in at least some of all possible executions.
This is the case if the minimum age of the block prior to the access is guaranteed to be less than the cache's associativity.
Similarly, the EM analysis can be used to determine that an access results in a miss in at least some of the possible executions.

Combining the results of the two analyses, some accesses can be classified as ``definitely unknown''. Then, further refinement by model checking is provably impossible.
Classifications as ``exists hit'' or ``exists miss'', which occur if either the EH or the EM analysis is successful but not both, are also useful to reduce further model-checking efforts:
e.g. in case of ``exists hit'' it suffices to determine by model checking whether a miss is possible to fully classify the access.


\section{Cache Analysis by Model Checking}
\label{sec:model_checking}
\emph{All proofs can be found in \ifthenelse{\boolean{conferenceversion}}{Appendix B of the technical report~\cite{technical_report}}{Appendix~\ref{sec:mc_abstract_model_proofs}}.} %
We have seen a new abstract analysis capable of classifying certain cache accesses as ``definitely unknown''.
The classical ``may'' and ``must'' analyses and this new analysis classify a (hopefully large) portion of all accesses as ``always hit'', ``always miss'', or ``definitely unknown''. But, due to the incomplete nature of the analysis the exact status of some blocks remains unknown. Our approach is summarized at a high level in Listing~\ref{lst_AI}. Functions \verb!May!, \verb!Must!, \verb!ExistsHit! and \verb!ExistsMiss! return the result of the corresponding analysis, whereas \verb!CheckModel! invokes the model checker (see Listing~\ref{lst_MC}). Note that a block that is not fully classified as ``definitely unknown'' can still benefit from the \emph{Exists Hit} and \emph{Exists Miss} analysis during the model-checking phase. If the AI phase shows that there exists a path on which the block is a hit (respectively a miss), then the model checker does not have to check the ``always miss'' (respectively ``always hit'') property.

\lstdefinestyle{pseudocode}{
  language=C,
  breaklines=true,
  basicstyle=\footnotesize\ttfamily,
  keywordstyle=\bfseries\color{blue!40!black},
  identifierstyle=\color{green!40!black},
  morekeywords={function, if, else, return},
  emph={block,exist_hit,exist_miss},
  emphstyle={\color{red!40!black}},
  commentstyle=\itshape\color{black},
}

\begin{lstlisting}[caption={Abstract-interpretation phase},label=lst_AI,style=pseudocode]
function ClassifyBlock(block) {
  if (Must(block))      //Must analysis classifies the block   
    return AlwaysHit; 
  else if (!May(block))    //May analysis classifies the block
    return AlwaysMiss; 
  else if (ExistHit(block) && ExistMiss(block)) 
    return DefinitelyUnknown; //DU analysis classifies the block
  else // Otherwise, we call the model checker 
    return CheckModel(block, ExistsHit(block), ExistsMiss(block));
}
\end{lstlisting}
\begin{lstlisting}[caption={Model-checking phase},label=lst_MC,style=pseudocode]
function CheckModel(block, exist_hit, exist_miss) {
  if (exist_hit) { //block can not always miss 
    if (CheckAH(block)) return AlwaysHit;
  }
  else if (exist_miss) { //block can not always hit 
    if (CheckAM(block)) return AlwaysMiss;
  } else { //AI phase did not provide any information
    if (CheckAH(block)) return AlwaysHit;
    else if (CheckAM(block)) return AlwaysMiss;
  }
  return DefinitelyUnknown;
}
\end{lstlisting}

We shall now see how to classify these remaining blocks using model checking.
Not only is the model-checking phase \emph{sound}, i.e. its classifications are correct, it is also \emph{complete} relative to our control-flow-graph model, i.e. there remain no unclassified accesses: each access is classified as ``always hit'', ``always miss'' or ``definitely unknown''.
Remember that our analysis is based on the assumption that each path is semantically feasible.

In order to classify the remaining unclassified accesses, we feed the model checker a finite-state machine modeling the cache behavior of the program, composed of
\begin{inparaenum}[i)]
\item a model of the program, yielding the possible sequences of memory accesses \item a model of the cache.
\end{inparaenum}
In this section, we introduce a new cache model, focusing on the state of a particular memory block to be classified,
which we further simplify using the results of abstract interpretation.

\newcommand{\focus}{\odot}
\newcommand{\focusdomain}{{C_\focus}}
\newcommand{\focusupdate}{update_\focus}
\newcommand{\focusalpha}{\alpha_\focus}

As explained in the introduction, it would be possible to directly encode the control-flow graph of the program, adorned with memory accesses, as one big finite-state system.
A first step is obviously to slice that system per cache set to make it smaller.
Here we take this approach further by defining a model sound and complete with respect to a given memory block $a$:
parts of the model that have no impact on the caching status of $a$ are discarded, which greatly reduces the model's size.
For each unclassified access, the analysis constructs a model focused on the memory block accessed, and queries the model checker.
Both the simplified program model and the focused cache model are derived automatically, and do not require any manual interaction.

The \emph{focused cache model} is based on the following simple property of LRU: a memory block is cached if and only if its age is less than the associativity $k$, or in other words, if there are less than $k$ younger blocks.
In the following, w.l.o.g., let $a \in M$ be the memory block we want to focus the cache model on.
If we are only interested in whether $a$ is cached or not, it suffices to track the set of blocks younger than $a$. 
Without any loss in precision concerning $a$, we can abstract from the relative ages of the blocks younger than $a$ and of those older than $a$.

Thus, the domain of the focused cache model is $\focusdomain = \mathcal{P}(M) \cup \{\varepsilon\}$.
Here, $\varepsilon$ is used to represent those cache states in which $a$ is not cached.
If $a$ is cached, the analysis tracks the set of blocks younger than $a$.
We can relate the focused cache model to the concrete cache model defined in Section~\ref{sec:cache_modeling} using an abstraction function mapping concrete cache states to focused ones:
       {\small\begin{align}
		\focusalpha:	C					& \rightarrow	\focusdomain \nonumber\\
					q					& \mapsto
		\begin{cases}
			\varepsilon & \text{ if } q(a) \geq k\\
			\{b \in M \mid q(b) < q(a)\} & \text{ if } q(a) < k
		\end{cases}
	\end{align}}

The focused cache update~$\focusupdate$ models a memory access as follows:
	{\small\begin{align}
		\focusupdate: \focusdomain \times M &	 \rightarrow \focusdomain\nonumber\\
		(\widehat{Q}, b)						&	\mapsto 
		\begin{cases}
			\emptyset			& \text{ if } b = a\\
			\varepsilon		& \text{ if } b \neq a \wedge \widehat{Q} = \varepsilon\\
			\widehat{Q} \cup \{b\}	& \text{ if } b \neq a \wedge \widehat{Q} \neq \varepsilon \wedge |\widehat{Q} \cup \{b\}| < k\\
			\varepsilon		& \text{ if } b \neq a \wedge \widehat{Q} \neq \varepsilon  \wedge |\widehat{Q} \cup \{b\}| = k\\
		\end{cases}
	\end{align}}
\todo{note that the abstract update has been slightly simplified syntactically by merging the 3rd and 4th case}
Let us briefly explain the four cases above. If $b=a$ (case 1), $a$ becomes the most-recently-used block and thus no other blocks are younger.
If $a$ is not in the cache and it is not accessed (case 2), then $a$ remains outside of the cache.
If another block is accessed, it is added to $a$'s younger set (case 3) unless the access causes $a$'s eviction, because it is the $k^{th}$ distinct younger block (case 4).

\paragraph{Example.} Figure~\ref{fig:focused} depicts a sequence of memory accesses and the resulting concrete and focused cache states (with a focus on block $a$) starting from an empty cache of associativity 2. We represent concrete cache states by showing the two blocks of age $0$ and $1$. The example illustrates that many concrete cache states may collapse to the same focused one. 
At the same time, the focused cache model does not lose any information about the caching status of the focused block, which is captured by the following lemma and theorem.
\begin{figure}[t]
\begin{center}
	\begin{tikzpicture}[scale=0.9,transform shape]
            \node[pp] (x) {};
            \node[pp] (y) [right=10mm of x] {};
            \node[pp] (a) [right=10mm of y] {};
            \node[pp] (v) [right=10mm of a] {};
            \node[pp] (w) [right=10mm of v] {};
            \node[pp] (l) [right=10mm of w] {};
        
            \ppedge{x}{y}{\lstinline[style=alang]|x|}
            \ppedge{y}{a}{\lstinline[style=alang]|y|}
            \ppedge{a}{v}{\lstinline[style=alang]|a|}
            \ppedge{v}{w}{\lstinline[style=alang]|v|}
            \ppedge{w}{l}{\lstinline[style=alang]|w|}
        

            \node[blue,above=2mm of x, yshift=0mm] (firstconcrete) {$[-,-]$};
            \node[red,below=2mm of x, yshift=0mm, text height=2mm] (firstfocused) {$\epsilon$};
            
            \node[blue,above=2mm of y, yshift=0mm] {$[x,-]$};
            \node[red,below=2mm of y, yshift=0mm, text height=2mm] {$\epsilon$};

            \node[blue,above=2mm of a, yshift=0mm] {$[y,x]$};
            \node[red,below=2mm of a, yshift=0mm, text height=2mm] {$\epsilon$};
            
            \node[blue,above=2mm of v, yshift=0mm] {$[a,y]$};
            \node[red,below=2mm of v, yshift=0mm, text height=2mm] {$\emptyset$};
            
            \node[blue,above=2mm of w, yshift=0mm] {$[v, a]$};
            \node[red,below=2mm of w, yshift=0mm, text height=2mm] {$\{v\}$};
            
            \node[blue,above=2mm of l, yshift=0mm] {$[w, v]$};
            \node[red,below=2mm of l, yshift=0mm, text height=2mm] {$\epsilon$};
            
            \node[anchor=right, left=5mm of x, yshift=6mm] {Concrete cache model:};
            \node[anchor=right, left=5mm of x, yshift=-5.25mm] {Focused cache model:};
	\end{tikzpicture}
\end{center}
	\caption{Example: concrete vs. focused cache model.}\label{fig:focused}
\end{figure}
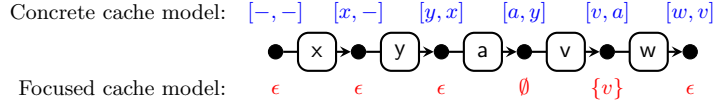

\begin{restatable}[Local Soundness and Completeness]{lemma}{MCUpdateConsistency}\label{lem:mc_update_function_consistency} The focused cache update abstracts the concrete cache update exactly:\todo{can we find a better formulation here?}
\begin{equation}
	\forall q \in C, \forall b \in M: \focusalpha(update(q, b)) = \focusupdate(\focusalpha(q), b).
\end{equation}
\end{restatable}
\todo{we might have to discuss this more}

The \emph{focused collecting semantics} is defined analogously to the \emph{collecting semantics} as the least solution to the following set of equations, where $R^C_\focus(v)$ denotes the set of reachable focused cache configurations at each program location, and $R^C_{\focus,0}(v) = \alpha^C_\focus(R^C_0(v))$ for all $v \in V$:
\begin{equation}\label{def:focused_concrete_reachable}
 	\forall v' \in V: R^C_\focus(v') = R^C_ {\focus,0}(v') \cup \bigcup_{(v,b,v') \in E} update^C_\focus(R^C_\focus(v), b),
\end{equation}
where $update^C_\focus$ denotes the focused cache update function lifted to sets of focused cache states, i.e., $update^C_\focus(Q, b) = \{\focusupdate(q, b) \mid q \in Q\}$, and $\alpha^C_\focus$ denotes the abstraction function lifted to sets of states, i.e., $\alpha^C_\focus(Q) = \{\alpha_\focus(q) \mid q \in Q\}$.


\begin{theorem}[Analysis Soundness and Completeness]
The focused collecting semantics is exactly the abstraction of the collecting semantics:
\begin{equation}
	\forall v \in V: \alpha^C_\focus(R^C(v)) = R^C_\focus(v).     
\end{equation}
\end{theorem}
\begin{proof}
From Lemma~\ref{lem:mc_update_function_consistency} it immediately follows that the lifted focused update $update^C_\focus$ exactly corresponds to the lifted concrete cache update $update^C$.

Since the concrete domain is finite, the least fixed point of the system of equations of Def.~\ref{def:concrete_reachable} is reached after a bounded number of Kleene iterations. 
One then just applies the consistency lemmas in an induction proof.\qed
\end{proof}

Thus we can employ the focused cache model in place of the concrete cache model without any loss in precision to classify accesses to the focused block as ``always hit'', ``always miss'', or ``definitely unknown''.

For the program model, we simplify the CFG without affecting the correctness nor the precision of the analysis:
\begin{inparaenum}[i)]
\item If we know, from may analysis, that in a given program instruction $a$ is never in the cache, then this instruction cannot affect $a$'s eviction: thus we simplify the program model by not including this instruction. 
\item When we encode the set of blocks younger than $a$ as a bit vector, we do not include blocks that the may analysis proved not to be in the cache at that location: these bits would anyway always be~$0$.
\end{inparaenum}



\section{Related Work}
\label{sec:related_work}
\todo{do we actually compare our results with those of the related work quantitatively? if not, move this section to after the experiments!?}

Earlier work by Chattopadhyay and Roychoudhury~\cite{Cha_rts13} refines memory accesses classified as ``unknown'' by AI using a software model-checking step:
when abstract interpretation cannot classify an access, the source program is enriched with annotations for counting conflicting accesses and run through a software model checker (actually, a bounded model checker).\todo{what is the consequence of using a bounded model checker in that case?}
Their approach, in contrast to ours, takes into account program semantics during the refinement step;
it is thus likely to be more precise on programs where many paths are infeasible for semantic reasons.
Our approach however scales considerably better, as shown in Section~\ref{sec:experimental_evaluation}:
not only do we not keep the program semantics in the problem instance passed to the model checker, which thus has finite state as opposed to being an arbitrarily complex program verification instance, we also strive to minimize that instance by the methods discussed in Section~\ref{sec:model_checking}.\todo{do we actually perform scalability comparisons?}

Chu et al.~\cite{Chu_rtas16} also refine cache analysis results based on program semantics, but by symbolic execution, where an SMT solver is used to prune infeasible paths.
We also compare the scalability of their approach to ours.

Our work complements \cite{Lv_rtss10}, which uses the classification obtained by classical abstract interpretation of the cache as a basis for WCET analysis on timed automata:
our refined classification would increase precision in that analysis.
Metta et al.~\cite{Met_lctes16} also employ model checking to increase the precision of WCET analysis. However, they do not take into account low-level features such as caches.

\section{Experimental Evaluation}
\label{sec:experimental_evaluation}



In industrial use for worst-case execution time, cache analysis targets a 
specific processor, specific cache settings, specific binary code loaded at a 
specific address.
The processor may have a hierarchy of caches and other peculiarities.
Loading object code and reconstructing a control-flow graph involves dedicated 
tools.
For data caches, a pointer value analysis must be run.
Implementing an industrial-strength 
analyzer including a pointer value analysis, or even interfacing in an existing 
complex analyzer, would greatly exceed the scope of this article.
For these reasons, our analysis applies to a single-level LRU instruction cache, and 
operates at LLVM bitcode level, each LLVM opcode considered as an elementary instruction.
This should be representative of analysis of machine code over LRU caches at a 
fraction of the engineering cost.

We implemented the classical may and must analyses, as well as our new definitely-unknown analysis and our conversion to model checking.
The model-checking problems are produced in the NuSMV format, then fed 
to  nuXmv~\cite{DBLP:conf/cav/CavadaCDGMMMRT14}.%
\footnote{\url{https://nuxmv.fbk.eu/}:
nuXmv checks for reachability using Kleene iterations over sets of 
states implicitly represented by binary decision diagrams (BDDs).
We also tried nuXmv's implementation of the IC3 algorithm with no speed improvement.}
We used an Intel Core i3-2120 processor (3.30~GHz) with 8~GiB RAM.

Our experimental evaluation is intended to show
\begin{inparaenum}[i)]
\item
precision gains by model checking (number of unknowns at the may/must stage vs. after the full analysis)
\item the usefulness of the definitely-unknown analysis 
(number of definitely-unknown accesses, which corresponds to the reduced number of MC calls, reduced MC cumulative execution time),
\item the 
global analysis efficiency (impact on analysis execution time, reduced number of MC calls).
\end{inparaenum}

\begin{figure}[t]
 \begin{center}
 \includegraphics[width=\sizefactor\textwidth]{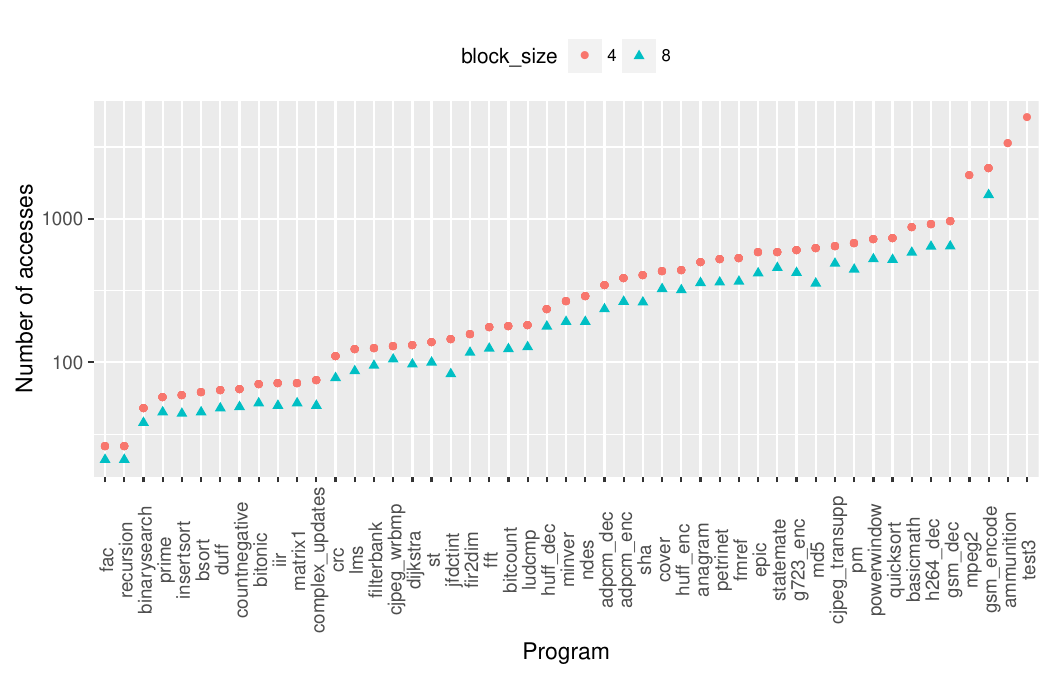}
 \end{center}
 \caption{Size of benchmarks in CFG blocks of 4 and 8 LLVM instructions.}
  \label{fig_benchs}
\end{figure}

As analysis target we use the 
TACLeBench benchmark suite~\cite{FalkWCET2016}\footnote{\url{http://www.tacle.eu/index.php/activities/taclebench}}, the successor of the M\"alardalen benchmark suite, which is commonly used in experimental evaluations of WCET analysis techniques. 
Figure~\ref{fig_benchs} (log. scale) gives the number of blocks in the control flow graph where a block is a sequence of instructions that are mapped to the same memory block.
In all experiments, we assume the cache to be initially empty and  we chose the following cache configuration: 8 instructions per block, 4 ways, 8 cache sets. \todo{this is very small, right?}
More details on the sizes of the benchmarks and further experimental results (varying cache configuration, detailed numbers for each benchmark,...) may be found in the technical report~\cite{technical_report}.

\newcommand{\nvspace}{\vspace{-6mm}}

\subsection{Effect of Model Checking on Cache Analysis Precision}

\begin{figure}[t]
\begin{center}
 \includegraphics[width=\sizefactor\textwidth]{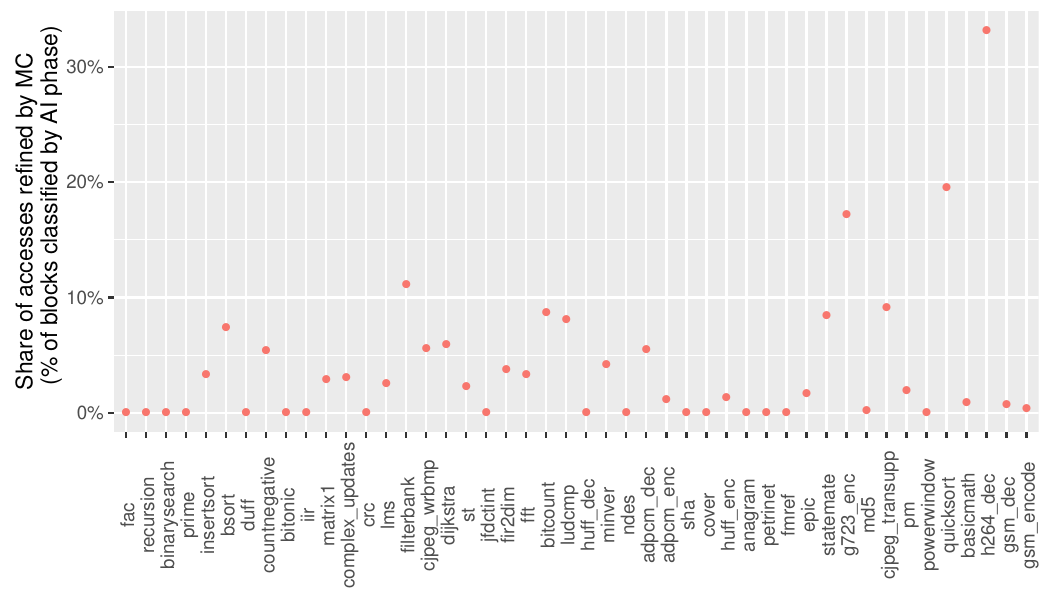}
 \end{center}
 \caption{Increase in hit/miss classifications due to MC relative to pure AI-based analysis.}
 \label{fig_imp}
\end{figure}

\todo{usually improvements are measured relative to previous results. so it would be more insightful to give relative improvements rather than absolute ones here. DONE: ratio}

\todo{percentage of additional hit and/or miss classifications over AI. if there are too many data points, maybe we can give geometric means over all benchmarks instead? Geom means would be zero because for some benchmark, improvement is null : MAY BE DONE; would have to look at ratio before/after, neither is zero}

Here we evaluate the improvement in the number of accesses 
classified as ``always hit'' or ``always miss''. 
%
In Figure \ref{fig_imp} we show by what percentage the number of such classifications increased from the pure AI phase due to model checking.\todo{important: what configuration are we talking about here? size of the cache, number of ways, block size... i copied the text below from even further below: does this apply to all the results in the whole section? then it should be mentioned earlier.DONE} 

As can be observed in the figure, more than 60\% of the benchmarks show an improvement and this improvement is greater than 5\% for 45\% of them.

We performed the same experiment under varying cache configurations (number of ways, number of sets, memory-block size) with similar outcomes.


\subsection{Effect{\nhspace}of{\nhspace}the{\nhspace}Definitely-Unknown{\nhspace}Analysis{\nhspace}on{\nhspace}Analysis{\nhspace}Efficiency}

\begin{figure}[t]
 \begin{subfigure}[Number of calls to the MC.\label{fig_MC_DU_sp}]{
   \includegraphics[width=0.45\textwidth]{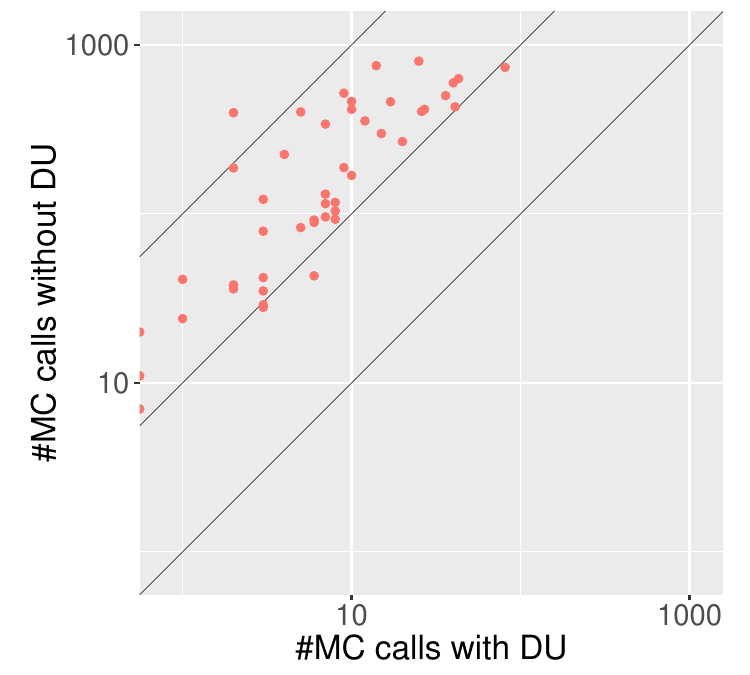}}
  \end{subfigure}
  \begin{subfigure}[Total MC time.\label{fig_time_DU_sp}]{
   \includegraphics[width=0.45\textwidth]{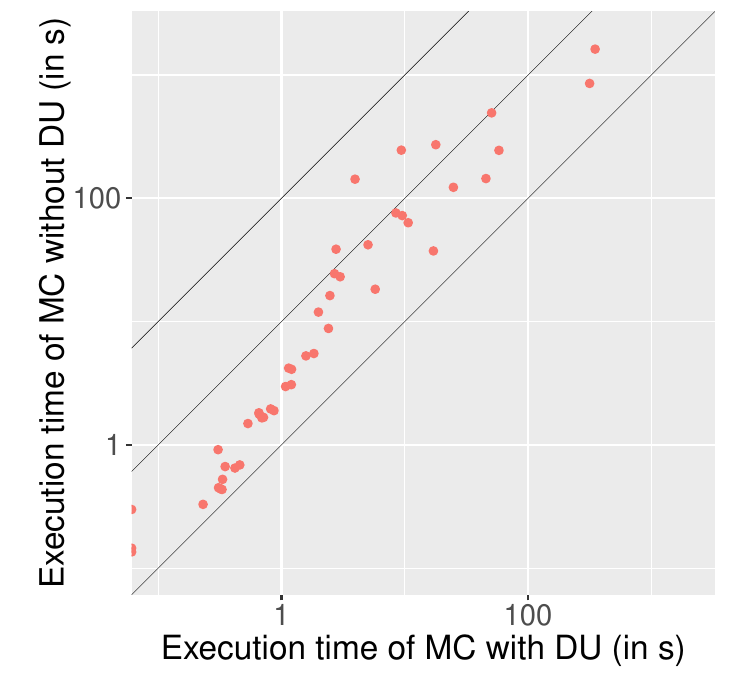}}
  \end{subfigure}
  \caption{Analysis efficiency improvements due to the definitely-unknown analysis.}
  \label{fig_DU}
\end{figure}

\todo{Geometric mean of the ratio (#calls without du/#calls with du) is 22.453 (zeroes are removed)}
\todo{must use logarithmic time scale! possibly scatter plot. otherwise you don't see anything for 2/3 of the figure. Logarithmic fig is available, as well as the scatter plot}

\begin{figure}[t]
 \begin{center}
 \includegraphics[width=\sizefactor\textwidth]{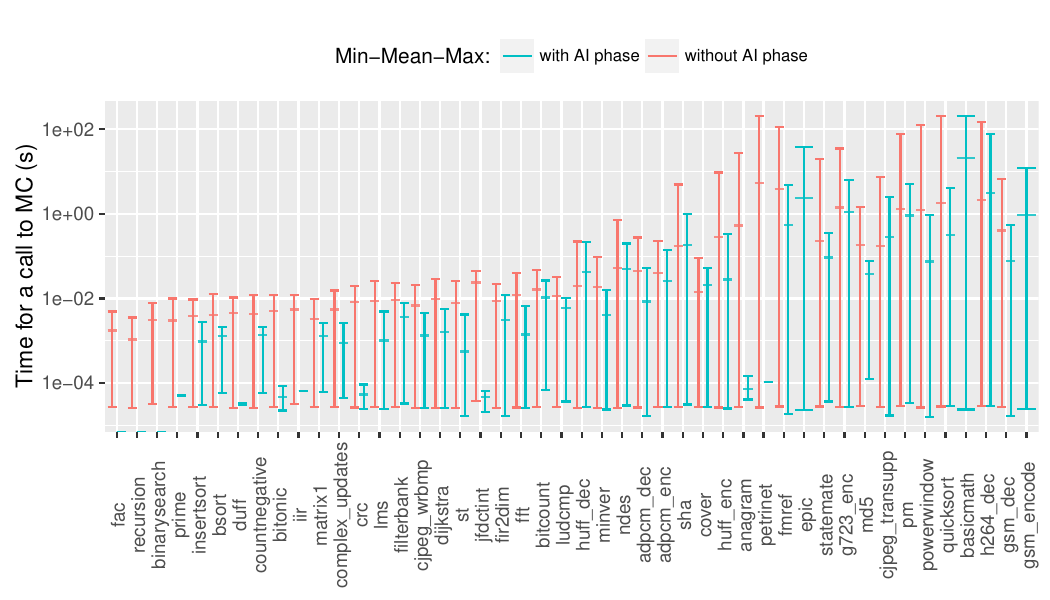}
 \end{center}
 \caption{MC execution time for individual call: min, mean, and max.}
 \label{fig_time_MCcalls_all_log}
\end{figure}

\todo{again, the relative improvement is more insightful. removing 10000 MC calls it not useful if 10000000 are left : MAY BE DONE}
\todo{$\exists$Hit or $\exists$Miss also reduce the number of MC calls is this quantified here? CAN NOT BE DONE :(}

We introduced the definitely-unknown analysis to reduce the number of MC calls: instead of calling the MC for each access not classified as either always hit or always miss by the classical static analysis, we also do not call it on definitely-unknown blocks. \todo{is this accurate? what about the $\exists$Hit, $\exists$Miss classifications?}
Figure~\ref{fig_MC_DU_sp} shows the number of MC calls with and without the definitely-unknown analysis.
The two lines parallel to the diagonal correspond to reductions in the number of calls by a factor of 10 and 100.
The definitely-unknown analysis significantly reduces the number of MC calls: for some of the larger benchmarks by around a factor of 100.
For the three smallest benchmarks, the number of calls is even reduced to zero: the definitely-unknown analysis perfectly completes the may/must analysis and no more blocks need to be classified by model checking. \todo{this is good and bad, because it also means that there is no potential for improvement...}
For 28 of the 46 benchmarks, fewer than 10 calls to the model checker are necessary after the definitely-unknown analysis.

This reduction of the number of calls to the model checker also results in significant improvements of the whole execution time of the analysis, which is dominated by the time spent in the model checker: see Figure~\ref{fig_time_DU_sp}.
On average (geometric mean) the total MC execution time is reduced by a factor of 3.7 compared with an approach where only the may and must analysis results are used to reduce the number of MC calls.
\todo{we should give the geo. mean of the time ratios. the cumulated model-checking time will be dominated by the largest benchmarks... The geometric mean of the ratio (execution time mc with DU/execution time mc without du) is 3.727 (zeroes are removed) DONE} 
\todo{so almost all benchmarks could be analyzed in less than 100 seconds even without DU!? regarding Figure~\ref{fig_time_DU_sp}...}
\todo{why is reduction in calls much larger than reduction in MC time? this would be interesting to discuss!}

Note that the definitely-unknown analysis itself is very fast: it takes less than one second on all benchmarks.
\todo{this used to say 16/46, 14/16, ... when i sum these up, i don't end up at 46: 16+14+1+1=32. what's going on? MERCI it was very wrong figures :p}

\subsection{Effect of Cache and Program Model Simplifications on Model-Checking Efficiency}

In all experiments we used the focused cache model: without this focused model, the model is so large that a timeout of one hour is reached for all but the 6 smallest benchmarks.
This shows a huge scalability improvement due to the focused cache model.
It also demonstrates that building a single model to classify all the accesses at once is practically infeasible.

Figure \ref{fig_time_MCcalls_all_log} shows the execution time of individual MC calls (on a log. scale) with and without program-model simplifications based on abstract-interpretation results. For each benchmark, the figure shows the maximum, minimum, and mean execution time of all MC calls for that benchmark.
We observe that the maximum execution time is always smaller with the use of the AI phase due to the simplification of program models.
Using AI results, there are fewer MC calls and many of the suppressed MC calls are ``cheap'' calls: this explains why the average may be larger with AI phase.
Some benchmarks are missing the ``without AI phase'' result: this is the case for benchmarks for which the analysis did not terminate within one hour.
\todo{why are we not comparing apples with apples, i.e., MC calls on the same set of queries?}

\begin{figure}[t]
 \begin{subfigure}[Number of calls to the MC.\label{fig_MCCalls_PA_sp}]{
   \includegraphics[width=0.45\textwidth]{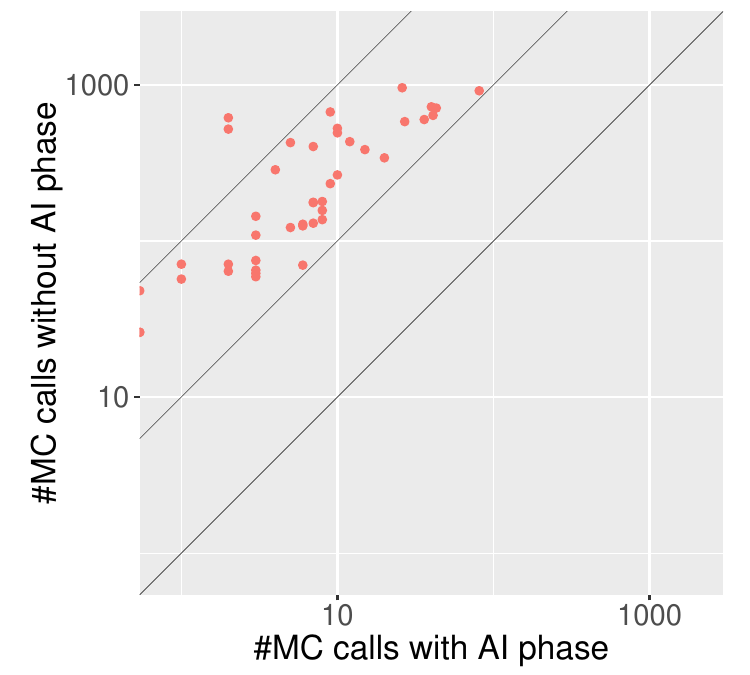}
  }
  \end{subfigure}
  \begin{subfigure}[Total MC time.\label{fig_MCtime_PA_sp}]{
    \includegraphics[width=0.45\textwidth]{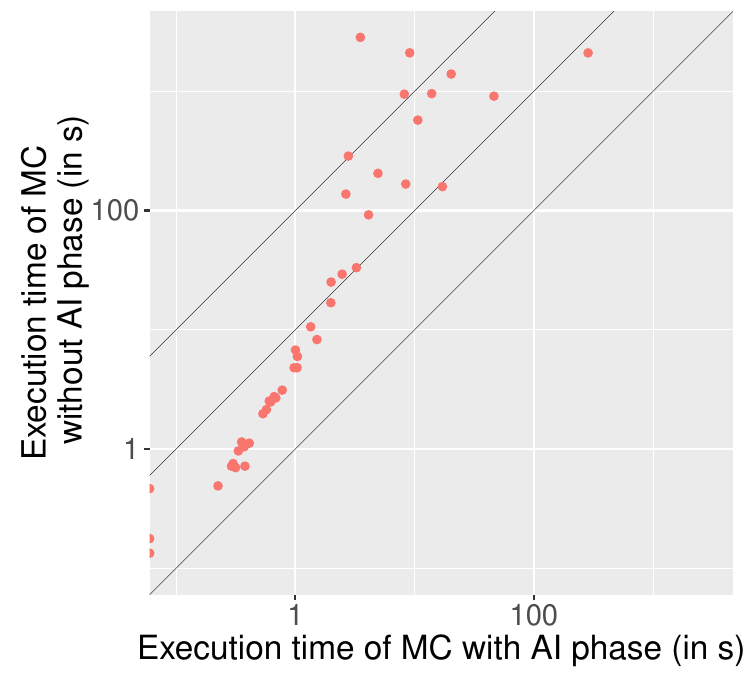}}
  \end{subfigure}
  \caption{Analysis efficiency improvements due to the entire AI phase.}
  \label{fig_PA}
\end{figure}

\subsection{Efficiency of the Full Analysis}
\todo{can we give the speedups in terms of total analysis time for each benchmark? MAY BE DONE}


First, we compare our approach to that of the related work~\cite{Cha_rts13,Chu_rtas16}.
Both tools from the related work operate at C level, while our analysis operates at LLVM IR level.
Thus it is hard to reasonably compare analysis precision.
To compare scalability we focus on total tool execution time, as this is available.
In the experimental evaluation of \cite{Cha_rts13} we see that it takes 395 seconds to analyze \benchmarkname{statemate} (they stop the analysis at 100 MC calls). With a 
similar configuration, 64 sets, 4 ways, 4 instructions per block (resp. 8 
instructions per blocks) our analysis makes 3 calls (resp. 0) to the model checker (compared with 832 (resp. 259) MC calls without the AI phase) and spends less than 3 seconds (resp. 1.5s) on the entire analysis.
Unfortunately, among all TACLeBench benchmarks \cite{Cha_rts13} gives scalability results only for \benchmarkname{statemate}, and thus no further comparison is possible.
The analysis from~\cite{Chu_rtas16} also spends more than 350 seconds
to analyze \benchmarkname{statemate}; for \benchmarkname{ndes} it takes 38 seconds
whereas our approach makes only 3 calls to the model checker and requires less than one second for the entire analysis.
This shows that our analysis scales better than the two related approaches.
However, a careful comparison of analysis precision remains to be done.

To see more generally how well our approach scales, we compare the total analysis time with and without the AI phase.
The AI phase is composed of the may, must and definitely-unknown analyses: 
without the AI phase, the model checker is called for each memory access and the program model is not simplified.\todo{ideally we should analyze the program simplification step separately, the focused cache model separately: DO NOT SCALE WITHOUT THE FOCUSED MODEL TIMEOUT ON ALMOST ALL BENCHMARKS, and the AI separately (split into the obvious may/must and the new part) MAY BE DONE}
\todo{i replaced these three commented lines above by the following, because the above don't seem to match the figures!!! OK}
On all benchmarks the number of MC calls is reduced by a factor of at least 10, sometimes exceeding a factor of 100 (see Figure~\ref{fig_MCCalls_PA_sp}).
This is unsurprising given the strong effect of the definitely-unknown analysis, which we observed in the previous section. Additional reductions compared with those seen in Figure~\ref{fig_MC_DU_sp} result from the classical may and must analysis.
Interestingly, the reduction in total MC time appears to increase with increasing benchmark sizes: see Figure~\ref{fig_MCtime_PA_sp}. While the improvement is moderate for small benchmarks that can be handled in a few seconds with and without the AI phase, it increases to much larger factors for the larger benchmarks.
\todo{Why is the total model checking time reduced by a smaller factor?}
\todo{Why is the reduction larger for larger benchmarks?}

It is difficult to ascertain the influence our approach would have on a full WCET analysis, with respect to both execution time and precision.
In particular, WCET analyses that precisely simulate the microarchitecture need to explore fewer pipeline states if fewer cache accesses are classified as ``unknown''. 
Thus a costlier cache analysis does not necessarily translate into a costlier analysis overall.
We consider a tight integration with a state-of-the-art WCET analyzer as interesting future work, which is beyond the scope of this paper.

\todo{no comparison of analysis time with and without model checking!?}

\todo{XXX comments
XXX Conclusion on the very good efficiency and scalability
what's a good terminology to distinguish static and dynamic accesses? 
maybe reference vs access? with and without model checking and with definitely 
unknown in between
references on predicted WCET (possibly security,...?) \todo{what happens if one 
includes a persistence analysis? can we do an exact MC-based persistence 
analysis---and according definitely-unknown analyses---as well?} 
improvements by path-sensitive analyses
of course, on the settings we actually implement)
\todo{we need to check how much novelty there is here} \todo{integration with 
WCET analysis suggests to only improve classification of accesses on worst-case 
execution path (may yield different efficiency); might be different for 
security questions}
(MC only for unclassified accesses)
definitely-unknown-AI+MC approach (MC only for (fewer) unclassified accesses)
definitely-unknown-AI+AI-assisted-MC approach (simplified MC for (fewer) 
unclassified accesses)
efficiency of overall WCET analysis?
runtimes:
runtimes
between cases where analysis proofs property and where it does not
definitely-unknown and MC as in the first figure
Geometric mean vs arithmetic mean.
}

\section{Conclusion and Perspectives}
We have demonstrated that it is possible to precisely classify all accesses to an LRU cache at reasonable cost by a combination of abstract interpretation, which classifies most accesses, and model checking, which classifies the remaining ones.
\todo{quickly recap possibly surprising possibility of definitely-unknown analysis? future work: apply this idea to regular value analysis, i.e., intervals, octagons, etc.}
\todo{future work: perfect persistence analysis? loop peeling?}

Like all other abstraction-interpretation-based cache analyses, at least those known to us, ours considers all paths within a control-flow graph to be feasible regardless of functional semantics.\todo{without the AI restriction this is contradictory considering the related work section}
Possible improvements over this include:
\begin{inparaenum}[i)]
\item encoding some of the functional semantics of the program into the model-checking problem \cite{Met_lctes16,Cha_rts13}
\item using ``trace partitioning'' \cite{Rival_Mauborgne_TOPLAS07} or ``path focusing'' \cite{Monniaux_Gonnord_SAS11} in the abstract-interpretation phase.
\end{inparaenum}

\clearpage
\bibliographystyle{splncs03}
\bibliography{cache_model_checking_analysis}

\ifthenelse{\boolean{conferenceversion}}{}{
\clearpage
\appendix
\section{Proofs for the Definitely-Unknown Abstract Interpretation}
\label{sec:abstract_interpretation_proofs}

As mention in the main part of this report, the soundness of the EH analysis comes from the consistency of the abstract transformers with respect to the concrete transformers. Here are the detailled proofs for the EH analysis.

\ConsistencyUpdateEH*
\begin{proof} Consistency of EH Analysis

	Let $(\hat{q}_0, \hat{q}_{\must0}) \in \mathcal{A}_{\existsh}$ and $b \in M$. We use the additional notations:
	\begin{itemize}
		\item $(\hat{q}_1, \hat{q}_{\must1}) = update_{\existsh}((\hat{q}_0, \hat{q}_{\must0}), b)$
		\item $\mathcal{Q}_1 = \gamma_{\existsh}(\hat{q}_1, \hat{q}_{\must1})$
		\item $\mathcal{Q}_0 = \gamma_{\existsh}(\hat{q}_0, \hat{q}_{\must0})$
		\item $\mathcal{Q}_2 = \{update^C(Q, b) \mid Q \in \mathcal{Q}_0\}$
	\end{itemize}
	We want to prove that $\mathcal{Q}_2 \subseteq \mathcal{Q}_1$.

	Let $\widetilde{Q}_2 \in \mathcal{Q}_2$ and $\widetilde{Q}_0 \in \mathcal{Q}_0$ such that $\widetilde{Q}_2 = update^C(\widetilde{Q}_0, b)$.

	Then:
	\[
		\forall \widetilde{q}_0 \in \widetilde{Q}_0, update(\widetilde{q}_0, b) \in \{update(q, b) \mid q \in \gamma_{\must}(\hat{q}_{\must0})\}
	\]

	Thus, the consistency of the must analysis gives us:
	\begin{multline}
		\forall \widetilde{q}_0 \in \widetilde{Q}_0, update(\widetilde{q}_0, b) \subseteq \gamma_{\must}(\{update_{\must}^a(\hat{q}_{\must0}, b)\}) = \gamma_{\must} (\hat{q}_{\must1})
	\end{multline}
	Then: $\widetilde{Q}_2 \subseteq \gamma_{\must} (\hat{q}_{\must1})$.

	To complete the proof that $\mathcal{Q}_2 \subseteq \mathcal{Q}_1$, it remains to prove that:
	\[
		\forall b' \in M, \exists q \in \widetilde{Q}_2 \text{, such that: }q(b') \leq \hat{q}_1(b')
	\]

	Let $b' \in M$. We have:
	\[
		\hat{q}_1(b') =
		\begin{cases}
			0		& \text{ if } b = b'\\
			\hat{q}_0(b')	& \text{ if } \hat{q}_{\must0}(b) \leq \hat{q}_0(b')\\
			\hat{q}_0(b')+1	& \text{ if } \hat{q}_{\must0}(b) > \hat{q}_0(b') \wedge \hat{q}_0(b') < k\\
			k		& \text{ if } \hat{q}_{\must0}(b) > \hat{q}_0(b') \wedge \hat{q}_0(b') = k\\
		\end{cases}
	\]

	Let $\widetilde{q}_0 \in \widetilde{Q}_0$ such that $\widetilde{q}_0(b') \leq \hat{q}_0(b')$. Let $\widetilde{q}_2 = update(\widetilde{q}_0, b) \in \widetilde{Q}_2$. We show that $\widetilde{q}_2$ is a good candidate, i.e. $\widetilde{q}_2(b') \leq \hat{q}_1(b')$.

	\[
		\widetilde{q}_2(b') =
		\begin{cases}
			0			& \text{ if } b = b'\\
			\widetilde{q}_0(b')	& \text{ if } \widetilde{q}_0(b') \geq \widetilde{q}_0(b)\\
			\widetilde{q}_0(b')+1	& \text{ if } \widetilde{q}_0(b') < \widetilde{q}_0(b) \wedge \widetilde{q}_0(b') < k\\
			k			& \text{ if } \widetilde{q}_0(b') < \widetilde{q}_0(b) \wedge \widetilde{q}_0(b') = k\\
		\end{cases}
	\]

	\begin{itemize}
		\item If $b = b'$:
			$\widetilde{q}_2(b') = 0 = \hat{q}_1(b')$.
		\item If $b \neq b' \wedge \hat{q}_{\must0}(b) \leq \hat{q}_0(b')$, then $\hat{q}_1(b') = \hat{q}_0(b')$.
			\begin{itemize}
				\item If $\widetilde{q}_0(b') \geq \hat{q}_{\must0}(b)$, then: $\widetilde{q}_0(b') \geq \widetilde{q}_0(b)$.

					Thus: $\widetilde{q}_2(b') = \widetilde{q}_0(b')$.

					Finally: $\hat{q}_1(b') = \hat{q}_0(b') \geq \widetilde{q}_0(b') = \widetilde{q}_2(b')$
				\item Otherwise, $\widetilde{q}_0(b') < \hat{q}_{\must0}(b)$ and thus: $\widetilde{q}_0(b') < \hat{q}_0(b')$.
					\begin{itemize}
						\item if $\widetilde{q}_0(b') \geq \widetilde{q}_0(b)$, then:

						  $\widetilde{q}_2(b') = \widetilde{q}_0(b') < \hat{q}_0(b') \leq \hat{q}_1(b')$
						\item if $\widetilde{q}_0(b') < \widetilde{q}_0(b)$, then:

						  $\widetilde{q}_0(b') < k$ and thus: $\widetilde{q}_2(b') = \widetilde{q}_0(b') + 1 \leq \hat{q}_0(b') \leq \hat{q}_1(b')$
					\end{itemize}
			\end{itemize}
		\item If $b \neq b' \wedge \hat{q}_{\must0}(b) > \hat{q}_0(b') \wedge \hat{q}_0(b') < k$, then $\hat{q}_1(b') = \hat{q}_0(b') + 1$.
			Moreover, $\widetilde{q}_0(b') \leq \hat{q}_0(b') < k$.

			Thus: $\widetilde{q}_2(b') \leq \widetilde{q}_0(b') +1 \leq \hat{q}_0(b') + 1 = \hat{q}_1(b')$

		\item Otherwise, $b \neq b' \wedge \hat{q}_{\must0}(b) > \hat{q}_0(b') \wedge \hat{q}_0(b') = k$. Then: $\hat{q}_1(b') = l$ and trivially: $\widetilde{q}_2(b') \leq \hat{q}_1(b')$
	\end{itemize}
	In every case, we have: $\widetilde{q}_2(b') \leq \hat{q}_1(b')$.

	Thus, $\widetilde{Q}_2 \in \mathcal{Q}_1$, proving that $\mathcal{Q}_2 \subseteq \mathcal{Q}_1$\qed
\end{proof}

\ConsistencyJoinEH*
\begin{proof}
	Let $((\hat{q}_1, \hat{q}_{M1}), (\hat{q}_2, \hat{q}_{M2})) \in \mathcal{A}_{\existsh}^2, Q_1 \in \gamma_{\existsh}(\hat{q}_1, \hat{q}_{M1}), Q_2 \in \gamma_{\existsh}(\hat{q}_2, \hat{q}_{M2})$. We use the additional notation:
	\begin{itemize}
		\item $Q_3 = Q_1 \cup Q_2$
		\item $(\hat{q}_3, \hat{q}_{M3}) = (\hat{q}_1, \hat{q}_{M1}) \sqcup_{\existsh} (\hat{q}_2, \hat{q}_{M2})$
	\end{itemize}
	We want to prove that: $Q_3 \in \gamma_{\existsh}(\hat{q}_3, \hat{q}_{M3})$.

	Let $b \in M$, $\displaystyle \min_{q \in Q_3}q(b) \leq \min_{q \in Q_1} q(b) \leq \hat{q}_1(b)$. Similarly, $\displaystyle \min_{q \in Q_3}q(b) \leq \hat{q}_2(b)$. Thus, $\displaystyle \min_{q \in Q_3}q(b) \leq \hat{q}_3(b)$.

	Then, using the consistency of the must join, we have: $Q_3 \in \gamma_{\existsh}(\hat{q}_3, \hat{q}_{M3})$\qed
\end{proof}

\SoundnessEHEM*
\begin{proof}
	Both the collecting semantics and the abstract EH semantics are defined as least solutions to sets of equations, i.e., least fixed points of functions corresponding to these equations.
	The two domains are both finite for a given program, as the number of memory blocks is finite.
	Thus, both domains have finite ascending chains, and so the least fixed points can be obtained in a finite number of Kleene iterations.

	Let $R^C_i : V \rightarrow \mathcal{P}(C)$ and $R_{\existsh,i} : V \rightarrow \mathcal{A}_{\existsh}$ denote the values reached in the $i^{th}$ Kleene iteration:
	\begin{align}
 		\forall v' \in V: R^C_{i+1}(v') & = R^C_0(v') \cup \bigcup_{(v,b, v') \in E} update^C(R_i^C(v), b),\label{localkleenecollecting}\\
		\forall v' \in V: R_{\existsh,i+1}(v') & = R_{\existsh,0}(v') \sqcup \bigsqcup_{(v,b,v') \in E} update_{\existsh}(R_{\existsh,i}(v), b).\label{localkleeneabstract}
	\end{align}
	We will prove by induction that for all $i \in \mathbb{N}$, we have $$\forall v' \in V: R^C_{i}(v') \in \gamma_{\existsh}(R_{\existsh,i}(v')).$$
	This then implies the theorem, as due to finite ascending chains, there is a $j \in \mathbb{N}$, such that the least solutions $R^C$ and $R_{\existsh}$ are $R^C_j$ and $R_{\existsh,j}$.

	Induction base ($i=0$): This follows immediately from to the assumption that $R_0^C(v) \in \gamma_{\existsh}(R_{\existsh,0}(v))$.\\
	Induction step ($i \rightarrow i+1$): Let $v' \in V$ be arbitrary.
	By induction hypothesis, we have $R_i^C(v) \in \gamma_{\existsh}(R_{\existsh,i}(v))$ for all $v \in V$ s.t. $(v,b,v') \in E$.
	By Lemma~\ref{lem:localconsistency} (local consistency) this implies $update^C(R_i^C(v), b) \in \gamma_{\existsh}(update_{\existsh}(R_{\existsh,i}(v), b))$ for all $b \in M$ and $v \in V$ s.t. $(v,b,v') \in E$.
	Applying Lemma~\ref{lem:joinconsistency} (join consistency) this in turn implies: \\$\bigcup_{(v,b, v') \in E} update^C(R_i^C(v), b) \in \gamma_{\existsh}(\bigsqcup_{(v,b,v') \in E)} update_{\existsh}(R_{\existsh,i}(v), b))$.\\ 
	Applying Lemma~\ref{lem:joinconsistency} again, as by assumption $R_0^C(v') \in \gamma_{\existsh}(R_{\existsh,0}(v'))$, yields:
	\begin{multline}
		R^C_0(v') \cup \bigcup_{(v,b, v') \in E} update^C(R_i^C(v), b)\\\in \gamma_{\existsh}(R_{\existsh,0}(v') \sqcup \bigsqcup_{(v,b,v') \in E} update_{\existsh}(R_{\existsh,i}(v), b)),
	\end{multline}
	which, by (\ref{localkleenecollecting}) and (\ref{localkleeneabstract}), is equivalent to $R_{i+1}^C(v') \in \gamma_{\existsh}(R_{\existsh,i+1}(v'))$.
	\qed
\end{proof}

In the case of the EM analysis, we are computing a safe lower bound on the maximal age of a block, thus the concretization is:
{\small\begin{align}
	\gamma_{\existsm}:	\mathcal{A}_{\existsm}		& \rightarrow	\mathcal{P}(\mathcal{P}(C))\nonumber\\
				(\hat{q}, \hat{q}_{\may})	& \mapsto 	\big\{Q \subseteq \gamma_{\may}(\hat{q}_{\may}) \mid \forall b \in M: \max_{q \in Q} q(b) \geq \hat{q}(b)\big\}
\end{align}}

Similarly to $update_{\existsh}$, the abstract transformer $update_{\existsm}((\hat{q}_{\existsm}, \hat{q}_{\may}), b)$ corresponding to an access to block $b$ is defined by the pair $(\hat{q}_{\existsm}',update_{\may}(\hat{q}_{\may}, b))$, where
{\small\begin{eqnarray}
\hat{q}_{\existsm}' = \lambda b'.
		\begin{cases}
			0		& \text{ if } b' = b\\
			\hat{q}(b')	& \text{ if } \hat{q}_{\may}(b) < \hat{q}(b')\\
			\hat{q}(b')+1	& \text{ if } \hat{q}_{\may}(b) \geq \hat{q}(b') \wedge \hat{q}(b') < k\\
			k		& \text{ if } \hat{q}_{\may}(b) \geq \hat{q}(b') \wedge \hat{q}(b') = k\\
		\end{cases}
\end{eqnarray}}

When joining two analysis states during the EM analysis, it is safe to take the maximum of the lower bounds on maximal ages. Indeed, as mentioned for the EH analysis, the union of reachable states over all incoming paths is reachable.
\begin{equation}
(\hat{q}_1, \hat{q}_{\may1}) \sqcup_{\existsm} (\hat{q}_2, \hat{q}_{\may2}) =
(\lambda b. \max(\hat{q}_1(b), \hat{q}_2(b)), \hat{q}_{\may1} \sqcup_{\may} \hat{q}_{\may2})
\end{equation}

The proof of EM analysis soundness is similar to the proof of EH analysis soundness and can be obtained by substituting the name of the transformers and the join. Thus, in the following, we only prove the consistency of the transformers and the join.

\begin{restatable}[Local Consistency]{lemma}{ConsistencyUpdateEM}
The abstract transformer $update_{\existsm}$ soundly approximates its concrete counterpart $update^C$:
\begin{multline}
		\forall (\hat{q}, \hat{q}_{\may}) \in \mathcal{A}_{\existsm}, \forall b \in M, \forall Q \in \gamma_{\existsm}(\hat{q}, \hat{q}_{\may}):\\
 update^C (Q, b) \in \gamma_{\existsm}(update_{\existsm}((\hat{q}, \hat{q}_{\may}), b)).
\end{multline}%
\end{restatable}
\begin{proof} Consistency of EM Analysis

	Let $(\hat{q}_0, \hat{q}_{\may0}) \in \mathcal{A}_{\existsm}$ and $b \in M$. We use the additional notations:
	\begin{itemize}
		\item $(\hat{q}_1, \hat{q}_{\may1}) = update_{\existsm}((\hat{q}_0, \hat{q}_{\may0}), b)$
		\item $\mathcal{Q}_1 = \gamma_{\existsm}(\hat{q}_1, \hat{q}_{\may1})$
		\item $\mathcal{Q}_0 = \gamma_{\existsm}(\hat{q}_0, \hat{q}_{\may0})$
		\item $\mathcal{Q}_2 = \{update^C(Q, b) \mid Q \in \mathcal{Q}_0\}$
	\end{itemize}
	We want to prove that $\mathcal{Q}_2 \subseteq \mathcal{Q}_1$.

	Let $\widetilde{Q}_2 \in \mathcal{Q}_2$ and $\widetilde{Q}_0 \in \mathcal{Q}_0$ such that $\widetilde{Q}_2 = update^C(\widetilde{Q}_0, b)$.

	Then:
	\[
		\forall \widetilde{q}_0 \in \widetilde{Q}_0, update(\widetilde{q}_0, b) \in \{update(q, b) \mid q \in \gamma_{\may}(\hat{q}_{\may0})\}
	\]

	Thus, the consistency of the may analysis gives us:
	\begin{multline}
		\forall \widetilde{q}_0 \in \widetilde{Q}_0, update(\widetilde{q}_0, b) \subseteq \gamma_{\may}(\{update_{\may}^a(\hat{q}_{\may0}, b)\}) = \gamma_{\may} (\hat{q}_{\may1})
	\end{multline}
	Then: $\widetilde{Q}_2 \subseteq \gamma_{\may} (\hat{q}_{\may1})$.

	To complete the proof that $\mathcal{Q}_2 \subseteq \mathcal{Q}_1$, it remains to prove that:
	\[
		\forall b' \in M, \exists q \in \widetilde{Q}_2 \text{, such that: }q(b') \leq \hat{q}_1(b')
	\]

	Let $b' \in M$. We have:
	\[
		\hat{q}_1(b') =
		\begin{cases}
			0		& \text{ if } b = b'\\
			\hat{q}_0(b')	& \text{ if } \hat{q}_{\may0}(b) < \hat{q}_0(b')\\
			\hat{q}_0(b')+1	& \text{ if } \hat{q}_{\may0}(b) \geq \hat{q}_0(b') \wedge \hat{q}_0(b') < k\\
			k		& \text{ if } \hat{q}_{\may0}(b) \geq \hat{q}_0(b') \wedge \hat{q}_0(b') = k\\
		\end{cases}
	\]

	Let $\widetilde{q}_0 \in \widetilde{Q}_0$ such that $\widetilde{q}_0(b') \geq \hat{q}_0(b')$. Let $\widetilde{q}_2 = update(\widetilde{q}_0, b) \in \widetilde{Q}_2$. We show that $\widetilde{q}_2$ is a good candidate, i.e. $\widetilde{q}_2(b') \geq \hat{q}_1(b')$.

	\[
		\widetilde{q}_2(b') =
		\begin{cases}
			0			& \text{ if } b = b'\\
			\widetilde{q}_0(b')	& \text{ if } \widetilde{q}_0(b') \geq \widetilde{q}_0(b)\\
			\widetilde{q}_0(b')+1	& \text{ if } \widetilde{q}_0(b') < \widetilde{q}_0(b) \wedge \widetilde{q}_0(b') < k\\
			k			& \text{ if } \widetilde{q}_0(b') < \widetilde{q}_0(b) \wedge \widetilde{q}_0(b') = k\\
		\end{cases}
	\]

	\begin{itemize}
		\item If $b = b'$:
			$\widetilde{q}_2(b') = 0 = \hat{q}_1(b')$.
		\item If $b \neq b' \wedge \hat{q}_{\may0}(b) < \hat{q}_0(b')$, then $\hat{q}_1(b') = \hat{q}_0(b')$.

			Moreover, $b \neq b' \Rightarrow \widetilde{q}_2(b') \geq \widetilde{q}_0(b') \geq \hat{q}_0(b') = \hat{q}_1(b')$.
		\item If $b \neq b' \wedge \hat{q}_{\may0}(b) \geq \hat{q}_0(b') \wedge \hat{q}_0(b') < k$, then $\hat{q}_1(b') = \hat{q}_0(b') + 1$.
			\begin{itemize}
				\item If $\widetilde{q}_0(b') \leq \hat{q}_{\may0}(b)$, then: $\widetilde{q}_0(b') \leq \widetilde{q}_0(b)$.

					Moreover, $b \neq b' \Rightarrow \widetilde{q}_0(b') < \widetilde{q}_0(b)$, and thus: $\widetilde{q}_2(b') \geq \hat{q}_0(b') + 1 = \hat{q}_1(b')$.
				\item Otherwise, $\widetilde{q}_0(b') > \hat{q}_{\may0}(b)$ and thus: $\widetilde{q}_0(b') > \hat{q}_0(b')$.
					\begin{itemize}
						\item if $\widetilde{q}_0(b') \geq \widetilde{q}_0(b)$, then:

						  $\widetilde{q}_2(b') = \widetilde{q}_0(b') \geq \hat{q}_0(b') +1 = \hat{q}_1(b')$
						\item if $\widetilde{q}_0(b') < \widetilde{q}_0(b) \wedge \widetilde{q}_0(b') < k$, then:

						  $\widetilde{q}_2(b') = \widetilde{q}_0(b') + 1 > \hat{q}_0(b') + 1 = \hat{q}_1(b')$
						\item if $\widetilde{q}_0(b') < \widetilde{q}_0(b) \wedge \widetilde{q}_0(b') = k$, then:

						  $\widetilde{q}_2(b') = k \geq \hat{q}_1(b')$
					\end{itemize}
			\end{itemize}
		\item If $b \neq b' \wedge \hat{q}_{\may0}(b) \geq \hat{q}_0(b') \wedge \hat{q}_0(b') = k$, then $\hat{q}_1(b') = k$.

			Thus: $b \neq b' \Rightarrow \widetilde{q}_2(b') \geq \widetilde{q}_0(b') \geq \hat{q}_0(b') = k \geq \hat{q}_1(b')$
	\end{itemize}
	In every case, we have: $\widetilde{q}_2(b') \geq \hat{q}_1(b')$.

	Thus, $\widetilde{Q}_2 \in \mathcal{Q}_1$, proving that $\mathcal{Q}_2 \subseteq \mathcal{Q}_1$\qed
\end{proof}

\begin{restatable}[Join Consistency]{lemma}{ConsistencyJoinEM}
	The join operator $\sqcup_{\existsm}$ is correct:
	\begin{multline}
			\forall ((\hat{q}_1, \hat{q}_{M1}), (\hat{q}_2, \hat{q}_{M2})) \in \mathcal{A}_{\existsm}^2, Q_1 \in \gamma_{\existsm}(\hat{q}_1, \hat{q}_{M1}), Q_2 \in \gamma_{\existsm}(\hat{q}_2, \hat{q}_{M2}):\\
			Q_1 \cup Q_2 \in \gamma_{\existsm}((\hat{q}_1, \hat{q}_{M1}) \sqcup_{\existsm} (\hat{q}_2, \hat{q}_{M2})).
	\end{multline}
\end{restatable}
\begin{proof}
	Let $((\hat{q}_1, \hat{q}_{M1}), (\hat{q}_2, \hat{q}_{M2})) \in \mathcal{A}_{\existsm}^2, Q_1 \in \gamma_{\existsm}(\hat{q}_1, \hat{q}_{M1}), Q_2 \in \gamma_{\existsm}(\hat{q}_2, \hat{q}_{M2})$. We use the additional notation:
	\begin{itemize}
		\item $Q_3 = Q_1 \cup Q_2$
		\item $(\hat{q}_3, \hat{q}_{M3}) = (\hat{q}_1, \hat{q}_{M1}) \sqcup_{\existsm} (\hat{q}_2, \hat{q}_{M2})$
	\end{itemize}
	We want to prove that: $Q_3 \in \gamma_{\existsm}(\hat{q}_3, \hat{q}_{M3})$.

	Let $b \in M$, $\displaystyle \max_{q \in Q_3}q(b) \geq \max_{q \in Q_1} q(b) \geq \hat{q}_1(b)$. Similarly, $\displaystyle \max_{q \in Q_3}q(b) \geq \hat{q}_2(b)$. Thus, $\displaystyle \max_{q \in Q_3}q(b) \geq \hat{q}_3(b)$.

	Then, using the consistency of the may join, we have: $Q_3 \in \gamma_{\existsm}(\hat{q}_3, \hat{q}_{M3})$\qed
\end{proof}

\section{Proofs for Reduced Models in Model Checking}
\label{sec:mc_abstract_model_proofs}

\MCUpdateConsistency*
\begin{proof}
	Let $q=[b_1,...,b_k] \in C$ a reachable cache state and $b \in M$ a memory block.

	We prove consistency by inspection of different possible cases. First half of the proof deals with cache states that do not contain the interesting block $a$ ($\forall i \in \llbracket 1, k\rrbracket, b_i \neq a$). Second half deals with cache states that contain it ($\exists i \in \llbracket 1, k \rrbracket, b_i = a$). In both part, we treat the cases where the block accessed $b$ is $a$ or not. Moreover, to treat $a$ eviction, the second half of the proof adds sub-cases for distinction of states containing $a$ ``at the end'' of the cache (near eviction).
\hspace{-2cm}
\begin{itemize}
	\item if $\forall i \in \llbracket 1,k \rrbracket, b_i \neq a$ (i.e. $a$ is not in $q$), then we can distinguish two cases:
	\begin{itemize}
		\item if $b = a$ then:
		\begin{align*}
			\focusalpha(update(q, b))	&= \focusalpha(update([\overset{\neq a}{b_1},\overset{\neq a}{b_2},...,\overset{\neq a}{b_k}], a)) \\
							&= \focusalpha([a, b_1, ..., b_{k-1}]) \\
							&= \{\}\\
			\focusupdate(\focusalpha(q), b)	&= \focusupdate(\focusalpha([\overset{\neq a}{b_1},\overset{\neq a}{b_2},...,\overset{\neq a}{b_k}]), a) \\
							&= \focusupdate(\varepsilon, a) \\
							&= \{\}
		\end{align*}
		So consistency holds.
		\item if $b \neq a$ then:
		{\small\begin{align*}
			\focusalpha(update(q, b))	&= \focusalpha(update([\overset{\neq a}{b_1}, \overset{\neq a}{b_2},...,\overset{\neq a}{b_k}], \overset{\neq a}{b})) \\
							&= \focusalpha([\overset{\neq a}{b'_1},\overset{\neq a}{b'_2},...,\overset{\neq a}{b'_k}]) \text{ where } \forall i \in \llbracket 1, k \rrbracket, b'_i \in \{b_1,...,b_k,b\} \\
							&= \varepsilon \\
			\focusupdate(\focusalpha(q), b)	&= \focusupdate(\focusalpha([\overset{\neq a}{b_1},\overset{\neq a}{b_2},...,\overset{\neq a}{b_k}]), b) \\
							&= \focusupdate(\varepsilon, \overset{\neq a}{b}) \\
							&= \varepsilon
		\end{align*}}
		So property consistency holds.
	\end{itemize}
	\item if $\exists i \in \llbracket 1,k \rrbracket$ such that $b_i = a$ ($a$ is in the cache), we also distinguish between the cases $b = a$ and $b \neq a$:
	\begin{itemize}
		\item if $b = a$ then:

		{\small\begin{align*}
			\focusalpha(update(q, b))	&= \focusalpha(update([\overset{\neq a}{b_1},...,\overset{\neq a}{b_{i-1}}, a, b_{i+1},...,b_k], a)) \\
							&= \focusalpha([a, b_1,...,b_{i-1}, b_{i+1},...,b_k]) \\
							&= \{\}\\
			\focusupdate(\focusalpha(q), b)	&= \focusupdate(\focusalpha([\overset{\neq a}{b_1},...,\overset{\neq a}{b_{i-1}}, a, b_{i+1},...,b_k]), a)\\
							&= \focusupdate(\{b_1,...,b_{i-1}\}, a) \\
							&= \{\}
		\end{align*}}
		So consistency holds.

		\item if $b \neq a$, there is different cases depending whether $b$ is in the cache before or after $a$ and depending if $a$ is the least recently used block:
		\begin{itemize}
			\item
			if there exists $j < i$ such that $b_j = b$ ($b$ is in the cache and is younger than $a$):
			\begin{align*}
				&\focusalpha(update(q, b)) \\
				&= \focusalpha(update([b_1,...,b_{j-1},b,b_{j+1},...,b_{i-1}, a,b_{i+1},...,b_k], b))\\
				&= \focusalpha([b,b_1,...,b_{j-1},b_{j+1},...,b_{i-1}, a,b_{i+1},...,b_k])\\
				&= \{b,b_1,...,b_{j-1},b_{j+1},...,b_{i-1}\}\\
				&= \{b_1,...,b_{i-1}\}
			\end{align*}
			\begin{align*}
				&\focusupdate(\focusalpha(q), b) \\
				&= \focusupdate(\focusalpha([b_1,...,b_{j-1},b,b_{j+1},...,b_{i-1},a, b_{i+1},...,b_k]), b)\\
				&= \focusupdate(\{b_1,...,b_{j-1},b,b_{j+1},...,b_{i-1}\}), b)\\
				&= \focusupdate(\{b_1,...,b_{i-1}\}, b)\\
				&= \{b_1,...,b_{i-1}\}
			\end{align*}
			So consistency holds.

			\item
			if there exists $j > i$ such that $b_j = b$ ($b$ is in the cache and is older than $a$):
			\begin{align*}
				&\focusalpha(update(q, b))\\
				&= \focusalpha(update([b_1,...,b_{i-1},a,b_{i+1},...,b_{j-1}, b,b_{j+1},...,b_k],b)) \\
				&= \focusalpha([b,b_1,...,b_{i-1},a,b_{i+1},...,b_{j-1}, b_{j+1},...,b_k], b) \\
				&= \{b,b_1,...,b_{i-1}\}
			\end{align*}
			\begin{align*}
				&\focusupdate(\focusalpha(q), b) \\
				&= \focusupdate(\focusalpha([b_1,...,b_{i-1},a,b_{i+1},...,b_{j-1}, b,b_{j+1},...,b_k]), b) \\
				&= \focusupdate(\{b_1,...,b_{i-1}\}, b) \\
				&= \{b,b_1,...b_{i-1}\}
			\end{align*}
			So consistency holds.

			\item
			if $\forall j, b_j \neq b$ and $i \neq k$ (i.e. $b$ is not in the cache and $a$ is not the least recently used block):

			\begin{align*}
				\focusalpha(update(q, b))	&= \focusalpha(update([b_1,...,b_{i-1}, a,b_{i+1},...,b_{k}], b)) \\
								&= \focusalpha([b,b_1,...,b_{i-1}, a,b_{i+1},...,b_{k-1}]) \\
								&= \{b,b_1,...,b_{i-1}\} \\
				\focusupdate(\focusalpha(q), b)	&= \focusupdate(\focusalpha([b_1,...,b_{i-1}, a,b_{i+1},...,b_{k}]), b) \\
								&= \focusupdate(\{b_1,...,b_{i-1},a\}, b) \\
								&= \{b,b_1,...,b_{i-1}\}
			\end{align*}
			So consistency holds.

			\item
			if $\forall j, b_j \neq b$ and $i = k$ (i.e. $b$ is not in the cache and $a$ is the least recently used block):

			{\small\begin{align*}
				\focusalpha(update(q, b))	&= \focusalpha(update([b_1,...,b_{k-1},a], b)) \\
								&= \focusalpha([b,b_1,...,b_{k-1}]) \\
								&= \varepsilon \\
				\focusupdate(\focusalpha(q), b)	&= \focusupdate(\focusalpha([b_1,...,b_{k-1},a]), b) \\
								&= \focusupdate(\{b_1,...,b_{k-1}\}, b) \\
								&= \varepsilon
			\end{align*}}
			So consistency holds.
		\end{itemize}
	\end{itemize}
\end{itemize}
In all cases, consistency holds.
\end{proof}

\section{Additional Results from Experimental Evaluation}

\label{app_eval}

\begin{figure}[t]
 \begin{center}
 \includegraphics[width=\sizefactor\textwidth]{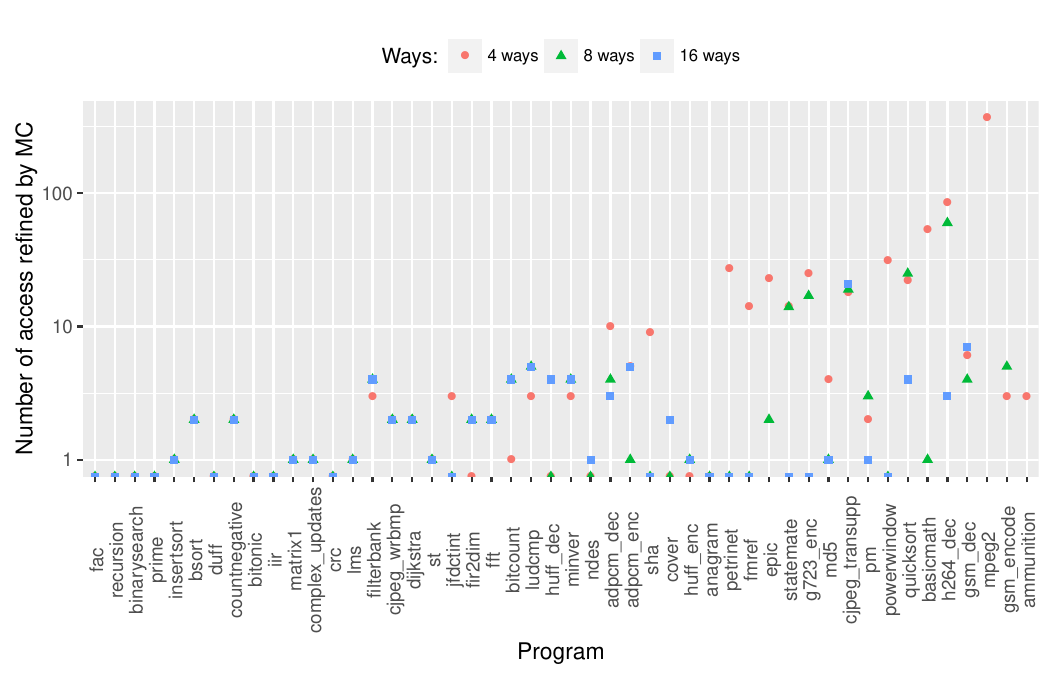}
 \end{center}
 \caption{Number of accesses whose classification was refined to hit or miss by model checking depending on the number of ways with 16 sets and 4 instructions per blocks.}
 \label{fig_ways}
\end{figure}

\begin{figure}[t] 
 \begin{center}
 \includegraphics[width=\sizefactor\textwidth]{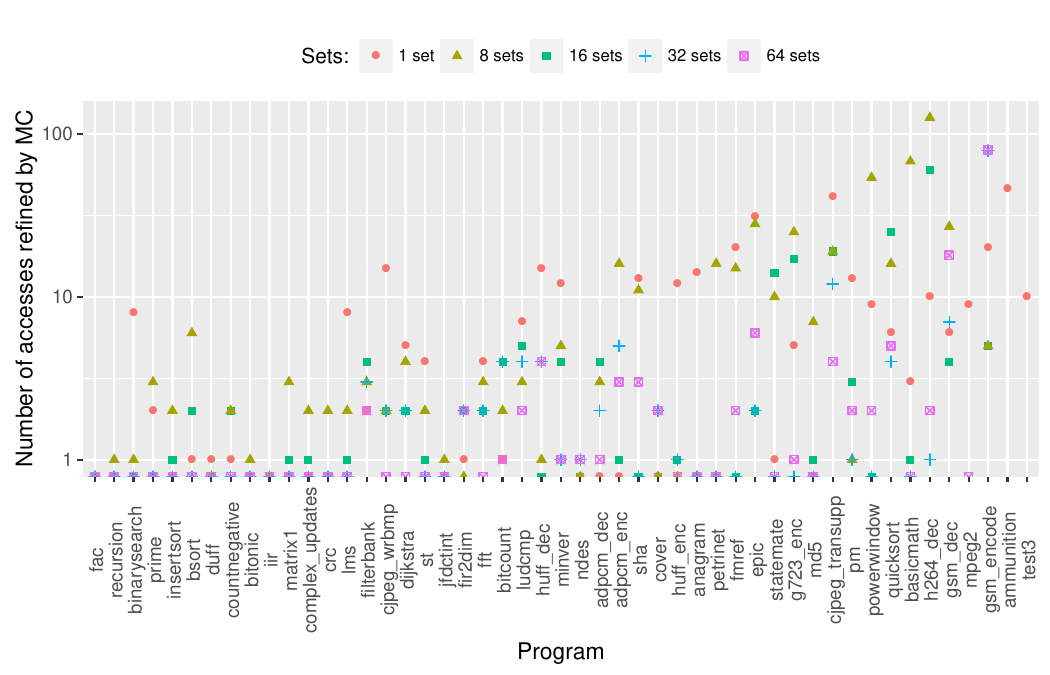}
 \end{center}
 \caption{Number of accesses whose classification was refined to hit or miss by model checking depending on the number of sets with 8 ways and 4 instructions per blocks.}
 \label{fig_sets}
\end{figure}

\begin{figure}[t]
 \begin{center}
 \includegraphics[width=\sizefactor\textwidth]{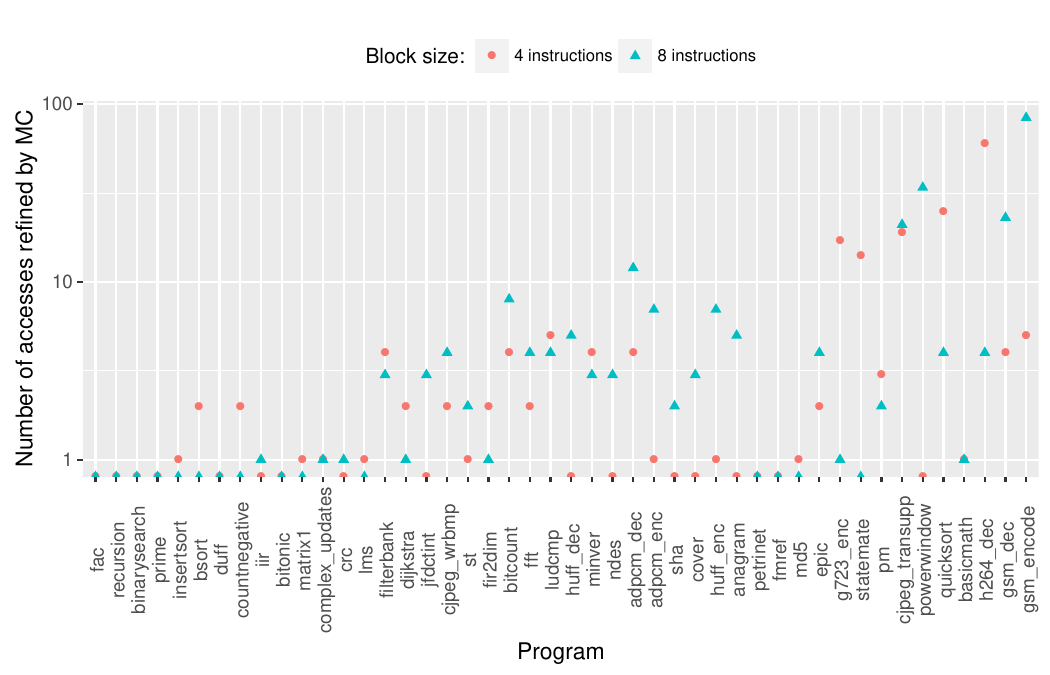}
 \end{center}
 \caption{Number of accesses whose classification was refined to hit or miss by model checking depending on the number of instructions per blocks with 8 ways and 16 sets.}
 \label{fig_blocks}
\end{figure}

\begin{figure}[ht]
 \begin{center}
 \includegraphics[width=\sizefactor\textwidth]{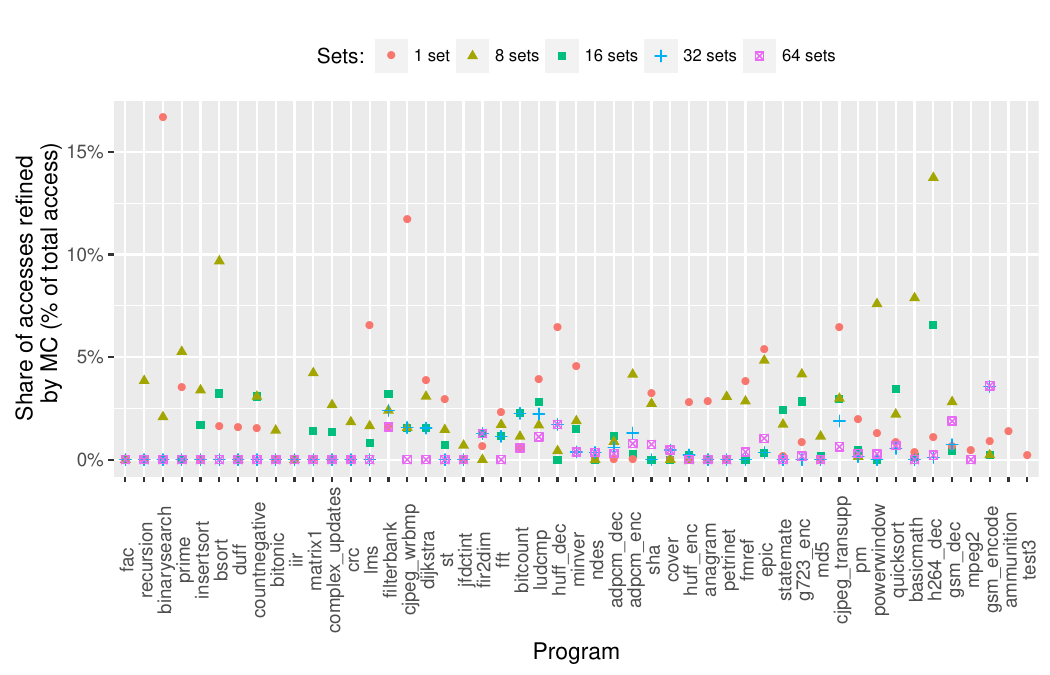}
 \end{center}
 \caption{Percentage of accesses whose classification was refined to hit or miss by model checking depending on the number of sets with 8 ways and 4 instructions per blocks.}
 \label{fig_rate_sets}
\end{figure}

\begin{figure}[ht]
 \begin{center}
 \includegraphics[width=\sizefactor\textwidth]{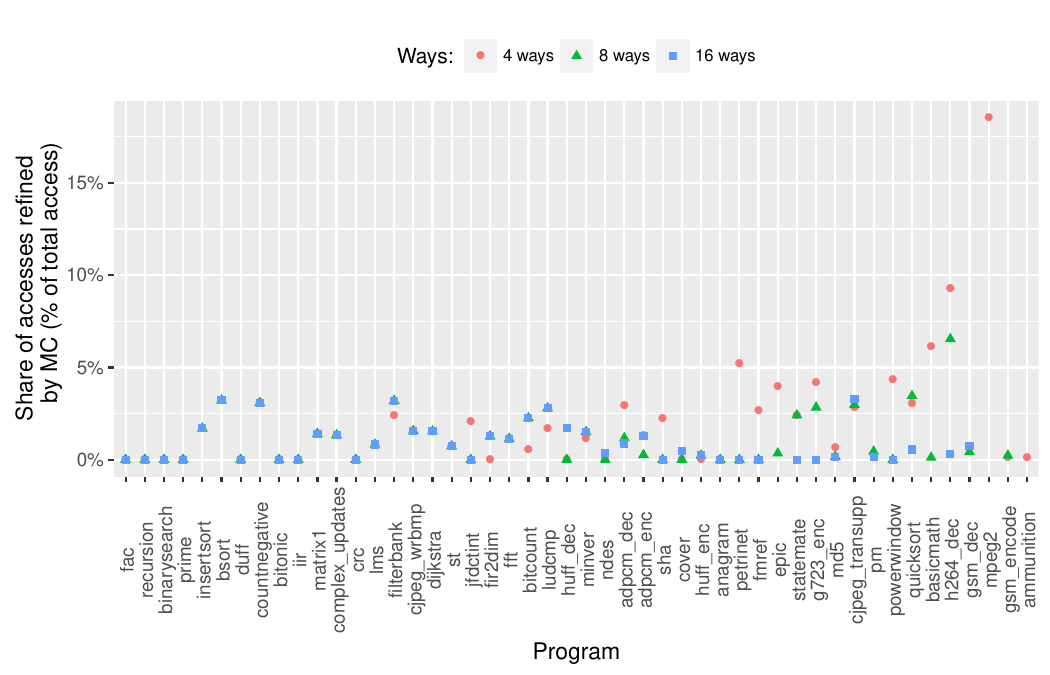}
 \end{center}
 \caption{Percentage of accesses whose classification was refined to hit or miss by model checking depending on the number of ways with 16 sets and 4 instructions per blocks.}
 \label{fig_rate_ways}
\end{figure}

\begin{figure}[ht]
 \begin{center}
 \includegraphics[width=\sizefactor\textwidth]{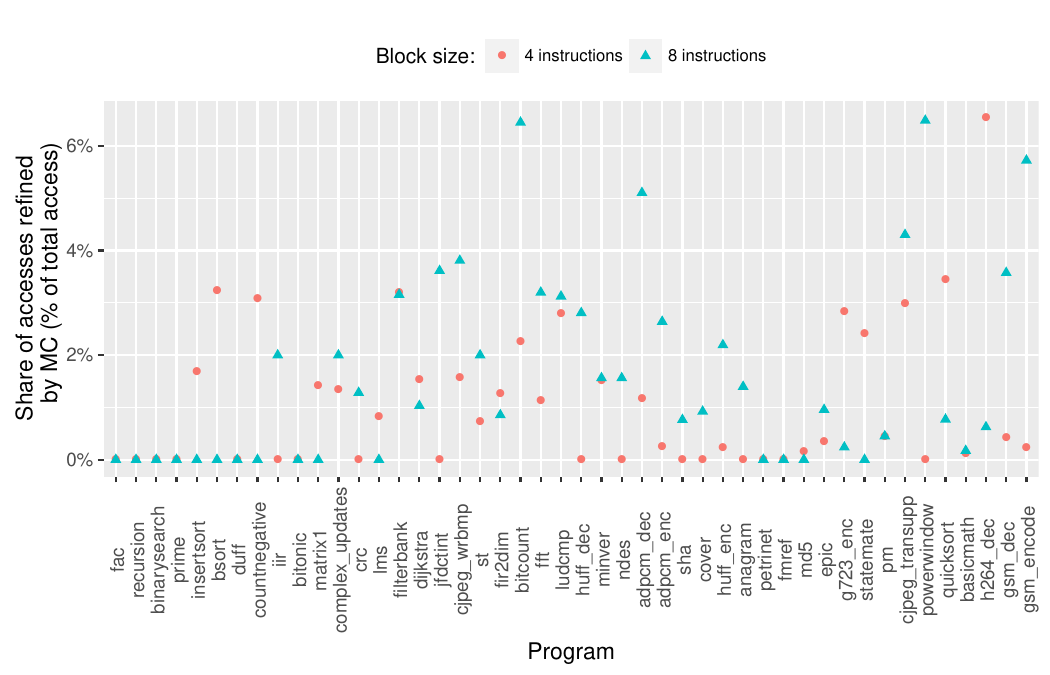}
 \end{center}
 \caption{Percentage of accesses whose classification was refined to hit or miss by model checking depending on the number of instructions per blocks with 8 ways and 16 sets.}
 \label{fig_rate_blocks}
\end{figure}


%

\begin{figure}[ht]
 \begin{center}
 \includegraphics[width=\sizefactor\textwidth]{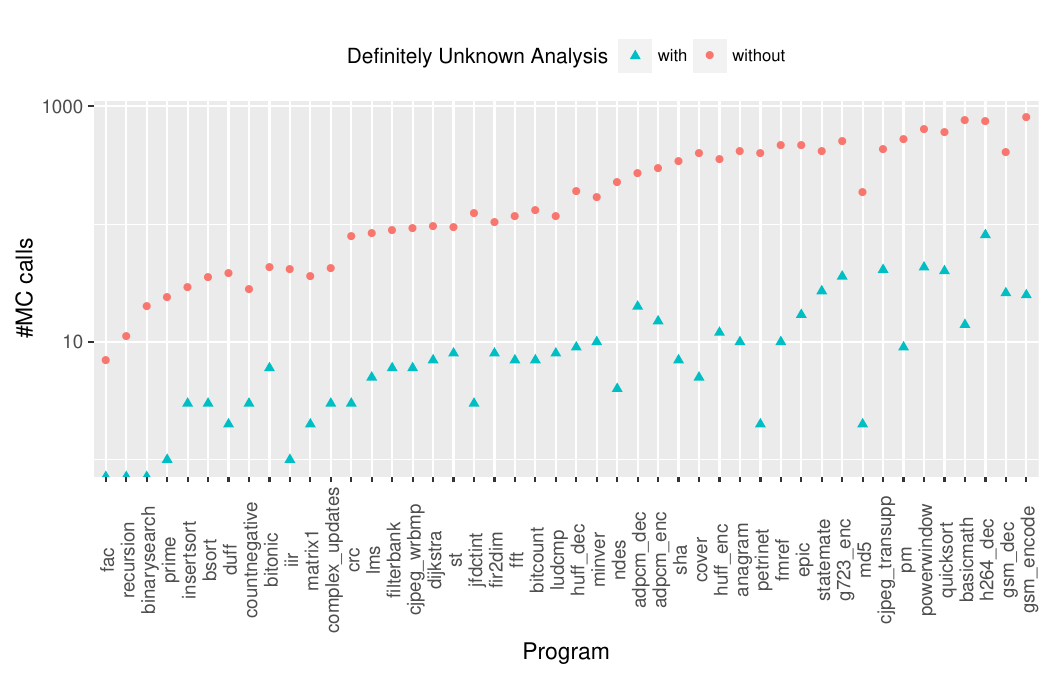}
 \end{center}
 \caption{Number of calls to the MC with and without the definitely-unknown analysis.}
 \label{fig_MCcalls_DU_log}
\end{figure}

\begin{figure}[ht]
 \begin{center}
 \includegraphics[width=\sizefactor\textwidth]{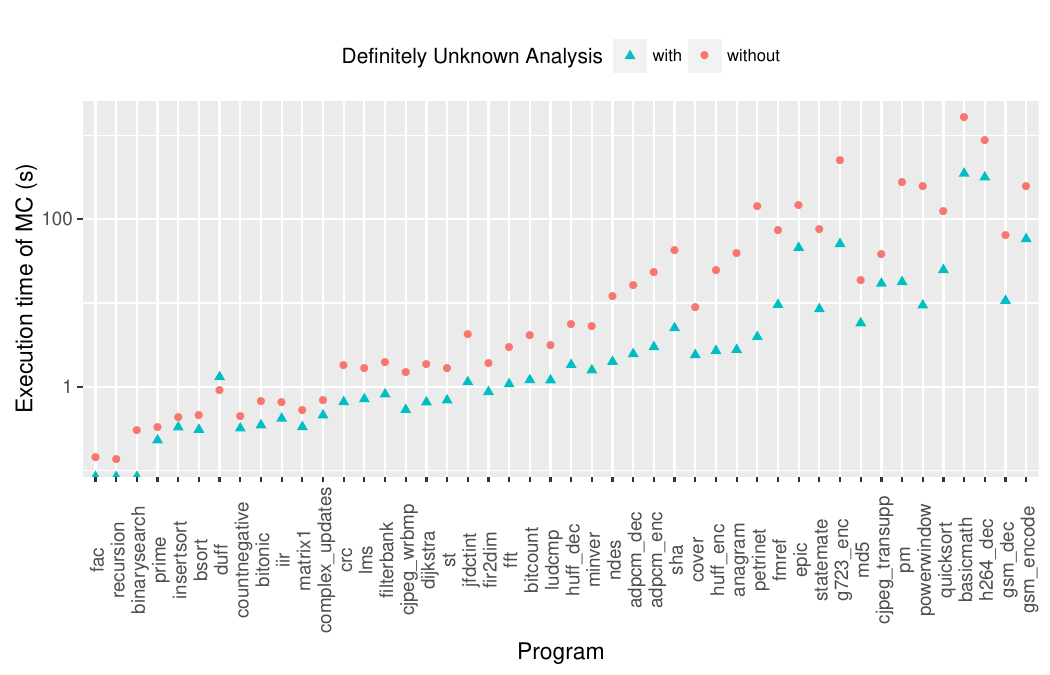}
 \end{center}
 \caption{Total MC time with and without the definitely-unknown analysis.}
 \label{fig_time_DU_log}
\end{figure}


\begin{figure}[t]
 \begin{center}
 \includegraphics[width=\sizefactor\textwidth]{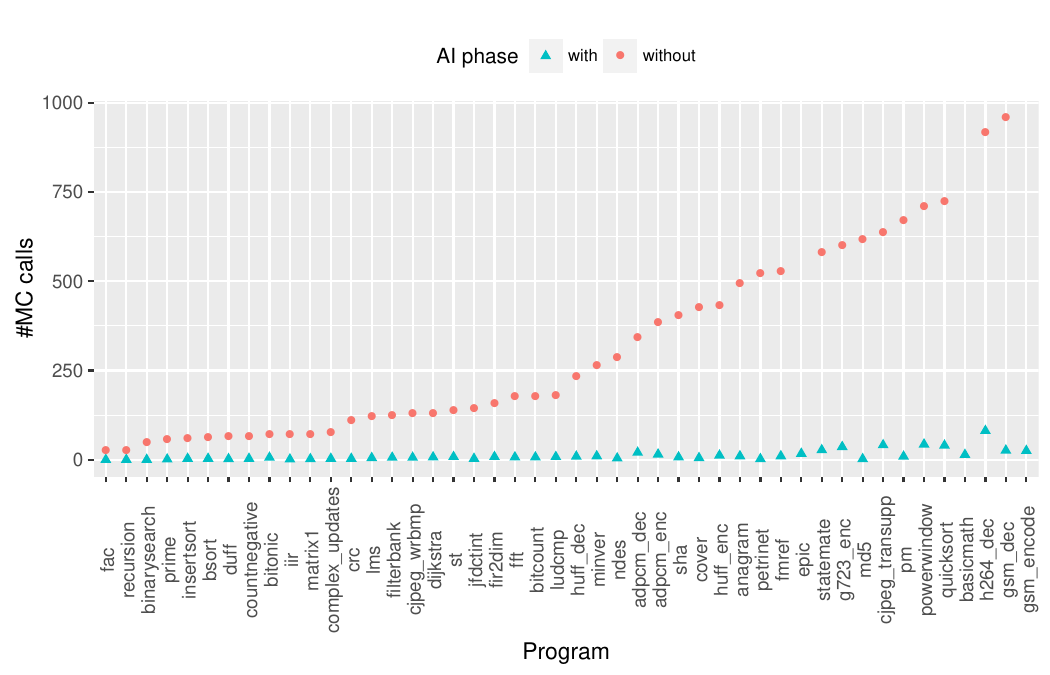}
 \end{center}
 \caption{Calls to the MC with and without AI phase (must/may + definitely unknown). \todo{how much of the reduction is due to the new AI? can we compare no analysis vs may/must, and then may/must vs may/must+def. unknown?}
 \todo{Geometric mean of the ratio (calls to mc without pre-analysis/with pre-analysis) is 31,466. Geom mean of the ratio (mc execution time without preanalysis/mc execution time with preanalysis) is 9.124}}
 \label{fig_MCcalls_PA}
\end{figure}


\begin{figure}[ht]
 \begin{center}
 \includegraphics[width=\sizefactor\textwidth]{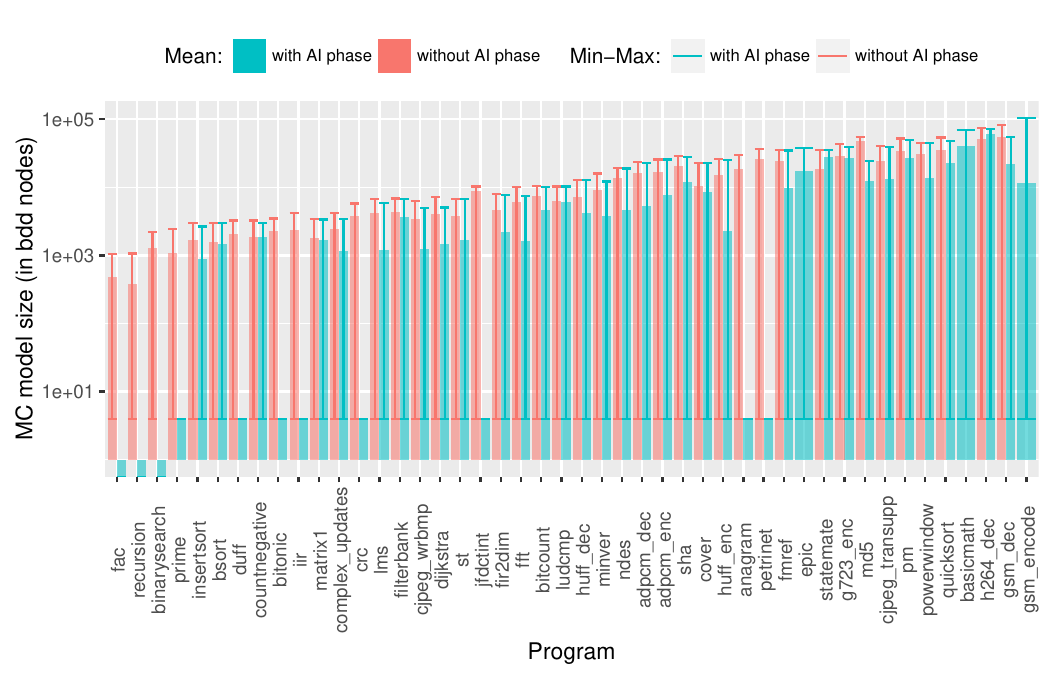}
 \end{center}
 \caption{Size of the MC model with and without the AI phase (must/may + definitely unknown).\todo{again what effect does DU have on MC model size?}}
 \label{fig_MCBDD_PA_log}
\end{figure}


In this section, we give further experimental results and more details on the results found in the main part of the paper. Please note that for all experiments any missing plot means that a timeout of one hour has been reached.

Figures \ref{fig_ways}, \ref{fig_sets}, \ref{fig_blocks} show the
number of accesses whose classification was refined to hit or miss by model checking under different cache configurations: varying the memory block size in Figure~\ref{fig_blocks}, the number of ways in Figure~\ref{fig_ways} 
and the number of cache sets in Figure~\ref{fig_sets}.
Figures \ref{fig_rate_sets}, \ref{fig_rate_ways}, \ref{fig_rate_blocks} show the share of the total number of accesses whose classification was refined by model checking for each benchmark.

Figures \ref{fig_MCcalls_DU_log} and \ref{fig_time_DU_log} show the usefulness of the definitely-unknown analysis in terms of the number of MC calls and the cumulative MC execution time.

Figure \ref{fig_MCcalls_PA} shows the effect of the AI phase on the number of MC calls.

Figure \ref{fig_MCBDD_PA_log} shows the size of the MC model with/without AI phase. We observe that the maximum is quite close with and without the AI phase. This influences the average: as there are fewer MC calls with AI phase, the average may thus be larger.}
\end{document}